\documentclass[aip,jmp,reprint,amssymb,amsmath,amsfonts,groupedaddress,onecolumn]{revtex4-1}
\usepackage{graphicx}
\usepackage{color}
\usepackage{colordvi}
\usepackage{xspace}
\usepackage{MnSymbol}

\usepackage{amsthm}

\usepackage{bbm}
\newcommand{\unity}{\mathbbmss{1}}
\newcommand{\N}{\ensuremath{\mathbb N}{}}
\newcommand{\R}{\ensuremath{\mathbb R}}

\newcommand{\C}{\ensuremath{\mathbb C}}

\newcommand{\ad}{\operatorname{ad}}

\newcommand{\abs}[1]{\ensuremath{\vert #1 \vert}}
\newcommand{\normone}[1]{\ensuremath{\Vert #1 \Vert{}}_1}
\newcommand{\normtwo}[1]{\ensuremath{\Vert #1 \Vert{}}_2}

\newcommand{\centre}{\mathfrak{center}}
\newcommand{\comm}{\mathfrak{com}}
\newcommand{\kernel}{\mathrm{ker}}

\newcommand{\uu}{\mathfrak{u}}
\newcommand{\su}{\mathfrak{su}}
\newcommand{\SU}{\operatorname{SU}}
\newcommand{\SL}{\operatorname{SL}}

\newcommand{\so}{\mathfrak{so}}

\newcommand{\usp}{\mathfrak{sp}}

\newcommand{\gl}{\mathfrak{gl}}

\newcommand{\fc}{\mathfrak{c}}

\newcommand{\fe}{\mathfrak{e}}
\newcommand{\ff}{\mathfrak{f}}
\newcommand{\fg}{\mathfrak{g}}
\newcommand{\fh}{\mathfrak{h}}
\newcommand{\ft}{\mathfrak{t}}

\newcommand{\fk}{\mathfrak{k}}

\newcommand{\fs}{\mathfrak{s}}

\newcommand{\Alt}{\mathrm{Alt}}
\newcommand{\Sym}{\mathrm{Sym}}

\newtheorem{theorem}{Theorem}
\newtheorem{lemma}[theorem]{Lemma}
\newtheorem{corollary}[theorem]{Corollary}
\newtheorem{proposition}[theorem]{Proposition}

\newtheorem{innercustomthm}{Theorem}
\newenvironment{customthm}[1]
  {\renewcommand\theinnercustomthm{#1}\innercustomthm}
  {\endinnercustomthm}
\theoremstyle{remark}
\newtheorem{definition}[theorem]{Definition}

\usepackage[bookmarks=false,pdfstartview={FitH}]{hyperref}

\begin{document}
% Use the \preprint command to place your local institutional report number 
% on the title page in preprint mode.
% Multiple \preprint commands are allowed.
%\preprint{}

\title{On squares of representations of compact Lie algebras} %Title of paper

% repeat the \author .. \affiliation  etc. as needed
% \email, \thanks, \homepage, \altaffiliation all apply to the current author.
% Explanatory text should go in the []'s, 
% actual e-mail address or url should go in the {}'s for \email and \homepage.
% Please use the appropriate macro for the type of information

% \affiliation command applies to all authors since the last \affiliation command. 
% The \affiliation command should follow the other information.

\author{Robert Zeier}
\email{robert.zeier@ch.tum.de}
\affiliation{Department Chemie, Technische Universit{\"a}t M{\"u}nchen,\\
Lichtenbergstrasse 4, 85747 Garching, Germany}
\author{Zolt{\'a}n Zimbor{\'a}s}
\email{zimboras@gmail.com}
\affiliation{Department of Computer Science, University College London,\\
Gower St, London WC1E 6BT, UK}
% Collaboration name, if desired (requires use of superscriptaddress option in \documentclass). 
% \noaffiliation is required (may also be used with the \author command).
%\collaboration{}
%\noaffiliation

%\date{\today}% It is always \today, today,
%             %  but any date may be explicitly specified
\date{24 August 2015}

\begin{abstract}
We study how tensor products of representations decompose when restricted from
a compact Lie algebra to one of its subalgebras. In particular, we are interested in
tensor squares which are tensor products of a representation with itself. We  
show in a classification-free manner that the sum of multiplicities and the sum of squares of 
multiplicities in the corresponding decomposition of a tensor square into irreducible representations 
has to strictly grow when restricted from a 
compact semisimple Lie algebra to a proper subalgebra. 
For this purpose, relevant details on tensor products of representations 
are compiled from the literature.
Since the sum of squares of multiplicities is equal to the dimension of 
the commutant of the tensor-square representation, it can be determined
by linear-algebra computations in a 
scenario where an
\emph{a priori} unknown Lie algebra is given by a set of generators which might not be a linear basis.
Hence,  our results offer a test to decide
if a subalgebra of a compact semisimple Lie algebra is a proper one
without calculating the relevant Lie closures,
which can be naturally applied in the field of controlled quantum systems.
\end{abstract}

%\pacs{02.30.Yy, 03.67.Ac, 02.20.Qs, 03.65.Aa}% insert suggested PACS numbers in braces on next line
%Control theory; 
%Quantum algorithms, protocols, and simulations; 
%General properties, structure, and representation of Lie groups
%Quantum systems with finite Hilbert space

%\keywords{control theory, fermionic systems, irreducible simple subalgebras, quantum computation, simulation of quantum systems, spin systems, representation theory}%

\maketitle %\maketitle must follow title, authors, abstract and \pacs

% Body of paper goes here. Use proper sectioning commands. 
% References should be done using the \cite, \ref, and \label commands

\section{Introduction}
The work of Dynkin~\cite{Dynkin57,Dynkin57b,Dynkin2000}  is a treasure trove of useful information
on representations of Lie algebras. In particular, Ref.~\onlinecite{Dynkin57} 
enumerates  all representations whose alternating square is irreducible (see Table~\ref{alt_square} below).
This classification triggered in Ref.~\onlinecite{ZS11} a study of tensor squares $\phi\otimes\phi$
which are defined for a representation $\phi$ of a Lie algebra $\fg$ as  the representation $(\phi \otimes \phi)(g):=
\phi(g) \otimes \unity_{\dim(\phi)} + \unity_{\dim(\phi)} \otimes \phi(g)$ with $g \in \fg$.
The tensor square $\phi\otimes\phi$ of the standard (i.e., defining) representation of
the Lie algebra $\su(\ell{+}1)$ corresponding to the special unitary group has the property that 
the dimension of its commutant $\comm[\phi\otimes\phi]$
has to grow when restricted to a proper subalgebra $\fh$, i.e.,
$\dim(\comm[(\phi\otimes\phi)|_{\fh}])>\dim(\comm[\phi\otimes\phi])=2$.
Here, $\comm[\psi]$  denotes the commutant of a representation $\psi$ of a Lie algebra
$\fg$ and consists of all complex matrices commuting with all $\psi(g)$ for $g\in\fg$. 
This discussion can be summarized as follows.

\begin{customthm}{A}[see Thm.~21 of Ref.~\onlinecite{ZS11}]\label{ThmA}
Given a subalgebra $\fh$ of $\su(\ell{+}1)$ with $\ell\geq 1$
and the standard representation $\phi$ of $\su(\ell{+}1)$, then
$\fh=\su(\ell{+}1)$ iff  $\dim(\comm[(\phi\otimes\phi)|_{\fh}])=2$.
\end{customthm}

The power of Theorem~\ref{ThmA} arises from the fact that its condition can be tested using only
a set of generators for the Lie algebra $\fh$, as those generators
are sufficient to compute the commutant $\comm[(\phi\otimes\phi)|_{\fh}]$ of the tensor square.
This led to control-theoretic applications in Ref.~\onlinecite{ZS11} 
where a controlled Schr{\"o}dinger equation provides certain initial
directions (i.e., Lie-algebra generators) in which a quantum system can be steered. 
In general, these directions do not linearly span but only generate
the \emph{a priori} unknown Lie algebra of all achievable directions, e.g., 
the three-dimensional, infinitesimal rotations around the $x$- and $y$-axes generate 
an additional infinitesimal rotation around the $z$-axis via the Lie commutator.
In this context,
one wants to decide effectively if 
the full Lie algebra $\su(\ell{+}1)$ is generated  without using the standard technique of computing the 
cumbersome Lie closure for a given set of generators. With the help Theorem~\ref{ThmA}, the question 
if a subalgebra of $\su(\ell{+}1)$ is a proper one can now be completely reduced to linear-algebra computations.
The proof given in Ref.~\onlinecite{ZS11} borrows heavily from the classification
of alternating squares in 
Dynkin's work\cite{Dynkin57} and treats all cases individually. 
The motivation of the current work
is to better understand the basic principles on which Theorem~\ref{ThmA} relies.
In doing so, we identify the following generalization of Theorem~\ref{ThmA} where $\su(\ell{+}1)$ is
substituted by an arbitrary compact, semisimple Lie algebra $\fg$ and the standard representation
of $\su(\ell{+}1)$ is replaced with an arbitrary finite-dimensional, faithful representation of $\fg$.

\begin{customthm}{B} \label{ThmB}
Given a subalgebra $\fh$ of a compact semisimple Lie algebra $\fg$
and a finite-dimensional, faithful representation $\phi$ of $\fg$,
then $\fh = \fg$ iff $\dim(\comm[(\phi \otimes \phi)|_{\fh}]) = \dim(\comm[\phi \otimes \phi])$.
\end{customthm}

Moreover, the proof of Theorem~\ref{ThmB} does not rely on information from classifications
and highlights general properties of restricted representations for compact, semisimple Lie
algebras and beyond.
Limitations on potential generalizations to arbitrary compact Lie algebras will be discussed 
in Section~\ref{discussion}.
Theorems \ref{ThmA} and \ref{ThmB} 
can be naturally transferred to connected, compact semisimple Lie groups, as
their representations induce always a semisimple representation of the corresponding
compact semisimple Lie algebra. 
Let us note, however, that the above theorems are not trivial consequences
of  the representation theory of general compact groups; in particular, they do not hold 
for finite groups. We provide a counter-example for non-connected compact groups
(recall that finite groups are formally non-connected, zero-dimensional Lie groups):
Consider 
the set $\mathcal{M}=
\left\{ \pm \left( \begin{smallmatrix} 1 & 0 \\ 0 & 1 \end{smallmatrix} \right),  
\pm \left( \begin{smallmatrix} i & 0 \\ 0 & -i \end{smallmatrix} \right),  
\pm \left( \begin{smallmatrix} 0 & -1 \\ 1 & 0 \end{smallmatrix} \right),  
\pm \left( \begin{smallmatrix} 0 & i \\ i & 0 \end{smallmatrix} \right), 
\tfrac{1}{2} \left( \pm \left( \begin{smallmatrix} 1 & 0 \\ 0 & 1 \end{smallmatrix} \right)
\pm \left( \begin{smallmatrix} i & 0 \\ 0 & -i \end{smallmatrix} \right)
\pm \left( \begin{smallmatrix} 0 & -1 \\ 1 & 0 \end{smallmatrix} \right)
\pm \left( \begin{smallmatrix} 0 & i \\ i & 0 \end{smallmatrix} \right) \right)
\right\}
$ of $24$ matrices. One can easily check that the matrices of $\mathcal{M}$ form a group $H$.
As all generators in $\mathcal{M}$ are contained in $\SU(2)$ (given in its standard representation), 
$H$ is a proper subgroup of $\SU(2)$. 
Moreover, $H$  is given here in a particular unitary representation,
is a double cover 
of the tetrahedral group, and is isomorphic to the special linear group $\SL(2,\mathbb{F}_3)$ of
$2\times 2$-matrices with entries from the finite field $\mathbb{F}_3$ and with determinant one. 
Denoting the standard representation of $\SU(2)$ by $\phi$, the described representation 
of $H$ will be naturally identified as $\phi|_{H}$. 
Let us consider the tensor squares $(\phi \otimes \phi)|_{H}$  and 
$\phi \otimes \phi$ which are defined 
as $(\phi \otimes \psi)(X):=
\phi(X) \otimes \psi(X)$ for elements $X$ of a group $G$. One obtains that $\dim(\comm[(\phi\otimes\phi)|_{H}])=2= \dim(\comm[(\phi\otimes\phi)])$ and shows that 
Theorems \ref{ThmA}  and \ref{ThmB} cannot be generalized to general compact groups which might not be connected.
Let us also note that finite subgroups $H$ of $\SU(d)$ with the property
that $\dim(\comm[(\phi\otimes\phi)|_{H}])= \dim(\comm[(\phi\otimes\phi)])$
are known as \emph{group designs} (which are particular types of unitary 2-designs) and
have also been studied in the context of quantum information theory.\cite{GAE07,RS09}

It is not unusual that limiting compact groups to connected ones
leads to significant modifications from a representation-theoretic point of view.
For example, it is well-known that there exist non-isomorphic compact groups 
with isomorphic representation rings.\cite{Z06b} On the other hand, two 
\emph{connected} compact groups can only have isomorphic representation rings if the corresponding groups 
are isomorphic.\cite{McMullen84,Handelman93, KLV14} In the case of a connected semisimple (complex)
Lie group, it is even enough to determine the so-called \emph{dimension datum} 
of a finite-dimensional, faithful 
representation~$\phi$ in order to fix its Lie algebra.\cite{LP90,Larsen04,AYY13}
The dimension datum corresponds roughly to knowing all dimensions for
representations occurring in any tensor power of $\phi$.
In this context, it  is surprising that the conditions in Theorems~\ref{ThmA} and \ref{ThmB}
rely only on the tensor square (but admittedly for a weaker conclusion).
Finally, Coquereaux and Zuber\cite{CZ11} proved
properties for the sum of multiplicities in the decomposition of a tensor product of two irreducible
representations of a simple Lie algebra
(see Appendix~\ref{CZ}), which also cannot be generalized to the non-connected group 
case.

We will assume that the reader has some familiarity with Lie algebras and representations,
which are both considered to be finite-dimensional throughout this work. All Lie algebras are defined
over the real or complex field, and all representations are matrix representations with
complex matrix entries.
We will use the words irreducible and simple 
(as well as completely reducible and semisimple)
as interchangeable names for properties of representations; irreducibility is always considered
with respect to the complex numbers.
For better accessibility, important facts and notations are recalled in Appendix~\ref{app_pre}. 
Our presentation will focus on compact Lie algebras, 
although the parallel language of complex reductive Lie algebras would 
be also suitable to state our results (cf.\ Appendix~\ref{app_pre});
and we will switch between them when necessary
without further comment. 

The article is organized as follows: We start in Section~\ref{Alt_Sym_Ten} by summarizing
a classification of representations with irreducible alternating and symmetric  squares;
the corresponding details are given in Appendices~\ref{Dynkin_techniques}
and \ref{case-by-case}.
The main classification-free results leading to Theorem~\ref{ThmB} are presented
in Section~\ref{classification-free}. We close by discussing generalizations
to general compact Lie algebras as well as lower bounds on the gap
between the dimensions of commutants of tensor square representations. 
Parts of the discussion are relegated to appendices
in order to streamline the presentation.

\section{Alternating, Symmetric, and Tensor Squares\label{Alt_Sym_Ten}}
In this section, we summarize results classifying representations
whose alternating and symmetric tensor squares are simple (i.e., irreducible). Streamlined proofs of these 
classifications which apply techniques developed by Dynkin (see Appendix~\ref{Dynkin_techniques} and 
Ref.~\onlinecite{Dynkin57})
are relegated to
Appendix~\ref{case-by-case}. 
The classification results allow us to prove techniques for distinguishing $\so(k)$, $\usp(\ell)$, or $\su(\ell{+}1)$
from its subalgebras
(cf.\ Refs.~\onlinecite{ZS11,ZZK14}), but we also provide simplified, classification-independent proofs
for $\so(k)$ and $\usp(\ell)$. By detailing the arguments for the cases
$\so(k)$, $\usp(\ell)$, or $\su(\ell{+}1)$, we also provide prototypes for
the general, classification-free proofs in Section~\ref{classification-free}.---We start with the classification
of simple alternating squares.
\begin{theorem}[Dynkin]\label{thm_alt_square}
Let $\phi$ denote a faithful representation of a compact semisimple Lie algebra
$\fg$ such that the alternating square $\Alt^2\phi$ is simple.
All possible cases (up to outer automorphisms of $\fg$) are given in Table~\ref{alt_square}.
\end{theorem}

{
\fontsize{10}{12}
\renewcommand{\baselinestretch}{1}
\selectfont
\begin{table}[tb]
\caption{\label{alt_square}Irreducible representations whose respective alternating square is also irreducible (Dynkin)}
\begin{ruledtabular}
\begin{tabular}[t]{l@{\hspace{4pt}}l@{\hspace{4pt}}l@{\hspace{4pt}}l@{\hspace{1pt}}l@{\hspace{4pt}}l@{\hspace{1pt}}l}
case & $\fg$ & $\ell$ & $\phi$  & dim$(\phi)$ & $\Alt^2 \phi$ & dim$(\Alt^2 \phi)$\\[1mm] \hline\\[-3.5mm]
(1a) & $\so(2\ell+1)$ & $\ell > 2$ &
$(1,0,\ldots,0)$ & $2\ell{+}1$ & $(0,1,0,\ldots,0)$  & $(2\ell{+}1)\ell$\\ 
(1b) & $\so(5)$ & -- & $(1,0)$ & $5$ & $(0,2)$ & $10$\\ 
(2a) & $\so(2\ell)$ & $\ell > 3$ &
$(1,0,\ldots,0)$ & $2\ell$ & $(0,1,0,\ldots,0)$ & $(2\ell{-}1)\ell$\\ 
(2b) & $\so(6)$ & -- & $(1,0,0)$ & $6$ & $(0,1,1)$ & $15$\\
(3) & $\su(\ell+1)$ & $\ell\geq 3$ &
$(0,1,0,\ldots,0)$ & $\tfrac{\ell(\ell+1)}{2}$ & $(1,0,1,0,\ldots,0)$ & $3 \tbinom{\ell+2}{4}$\\ 
(4a) & $\su(\ell+1)$ & $\ell>1$ &
$(2,0,\ldots,0)$ & $\tfrac{(\ell+1)(\ell+2)}{2}$ & $(2,1,0,\ldots,0)$ & $3 \tbinom{\ell+3}{4}$\\ 
(4b) & $\su(2)$ & -- &
$(2)$ & $3$ & $(2)$ & $3$\\ 
(5) & $\so(10)$ & -- &
$(0,0,0,1,0)$ & $16$ & $(0,0,1,0,0)$ & $120$\\ 
(6) & $\fe_6$ & -- &
$(1,0,0,0,0,0)$ & $27$ & $(0,0,1,0,0,0)$ & $351$\\ 
(7a) & $\su(\ell+1)$ & $\ell>1$ &
$(1,0,\ldots,0)$ & $\ell{+}1$ & $(0,1,0,\ldots,0)$ & $\tfrac{\ell(\ell+1)}{2}$\\ 
(7b) & $\su(2)$ & -- &
$(1)$ & $2$ & $(0)$ & $1$\\ 
\end{tabular}
\end{ruledtabular}
\end{table}
}

A streamlined proof of Theorem~\ref{thm_alt_square} is given in Appendix~\ref{AltSq}. This result implies the following theorem for distinguishing $\so(k)$ with $k \geq 5$ from its
subalgebras
(see also Thm.~15 in Ref.~\onlinecite{ZZK14}), but we also provide now a simplified proof relying on ideas from the proof of
Theorem~\ref{alt_self} in Appendix~\ref{case-by-case}.

\begin{theorem}\label{thm_so_square}
Given a subalgebra $\fh$ of $\so(k)$ with $k\geq 5$
and the standard representation $\phi:=\phi_{(1,0,\ldots,0)}$ of $\so(k)$,
 the following statements are equivalent:\\
(a) $\fh=\so(k)$.\\
(b) The representation $(\Alt^2\phi)|_{\fh}$ is simple.\\
(c) The representation $(\Alt^2\phi)|_{\fh}$ is simple and $(\Sym^2\phi)|_{\fh}$ splits into two simple
components. No simple component occurs more than once.\\
(d) The representation $(\phi\otimes\phi)|_{\fh}$ splits into three different simple components.\\
(e) The vector space of all complex matrices commuting with $(\phi\otimes\phi)|_{\fh}$ has dimension
three.
\end{theorem}

\begin{proof}
Note that $\Alt^2\phi$ is equivalent to the adjoint representation
and is  simple (see Lemma~\ref{adjoint} of Appendix~\ref{Dynkin_techniques}).
In particular, $\Alt^2\phi=\phi_{(0,2)}$ for $\so(5)$, $\Alt^2\phi=\phi_{(0,1,1)}$ for $\so(6)$, and
$\Alt^2\phi=\phi_{(0,1,0,\ldots,0)}$ for $k\geq 7$.
We have $\Sym^2\phi=\phi_{(2,0,\ldots,0)}\oplus\phi_{(0,\ldots,0)}$ (see Ex.~19.21 of Ref.~\onlinecite{FH91}).
It follows that both (b) and (c) are a consequence of (a). Obviously, (b) follows from (c).
The adjoint representation $\Alt^2\phi$ is no longer
simple when restricted to a proper subalgebra (see Lemma~\ref{adjoint} of
Appendix~\ref{Dynkin_techniques}) and (b) implies (a).
The statements (c) and (d) are equivalent as $(\Sym^2\phi)|_{\fh}$ splits into at least two components.
The equivalence of (d) and (e) follows from Lemma~\ref{reprtheory} of Appendix~\ref{app_pre}.
\end{proof}

Note that the proof relies critically on the irreducibility of the adjoint representation of $\so(k)$ for $k\geq 5$.---We present now the classification
of simple symmetric squares.

\begin{theorem}\label{thm_sym_square}
Let $\phi$ denote a faithful representation of a compact semisimple Lie algebra
$\fg$ such that the symmetric square $\Sym^2\phi$ is simple.
All possible cases (up to outer automorphisms of $\fg$) are given in Table~\ref{sym_square}.
\end{theorem}

{
\fontsize{10}{12}
\renewcommand{\baselinestretch}{1}
\selectfont
\begin{table}[tb]
\caption{\label{sym_square}Irreducible representations whose respective symmetric square is also irreducible}
\begin{ruledtabular}
\begin{tabular}[t]{l@{\hspace{4pt}}l@{\hspace{4pt}}l@{\hspace{4pt}}l@{\hspace{1pt}}l@{\hspace{4pt}}l@{\hspace{1pt}}l}
case & $\fg$ & $\ell$ & $\phi$  & dim$(\phi)$ & $\Sym^2 \phi$ & dim$(\Sym^2 \phi)$\\[1mm] \hline\\[-3.5mm]
(1) & $\usp(\ell)$ & $\ell \geq 1$ &
$(1,0,\ldots,0)$ & $2\ell$ & $(2,0,\ldots,0)$  & $(2\ell{+}1)\ell$\\ 
(2) & $\su(\ell+1)$ & $\ell\geq 1$ &
$(1,0,\ldots,0)$ & $\ell{+}1$ & $(2,0,\ldots,0)$ & $\tfrac{(\ell+1)(\ell+2)}{2}$\\ 
\end{tabular}
\end{ruledtabular}
\end{table}
}

A streamlined proof of Theorem~\ref{thm_sym_square} is given in Appendix~\ref{SymSq}. This result implies the following theorem for distinguishing $\usp(\ell)$ with $\ell \geq 2$ from its subalgebras, 
but we also provide now a simplified proof relying on ideas from the proof of
Theorem~\ref{sym_self} in Appendix~\ref{case-by-case}.

\begin{theorem}\label{thm_usp_square}
Given a subalgebra $\fh$ of $\usp(\ell)$ with $\ell\geq 2$
and the standard representation $\phi:=\phi_{(1,0,\ldots,0)}$ of $\usp(\ell)$,
 the following statements are equivalent:\\
(a) $\fh=\usp(\ell)$.\\
(b) The representation $(\Sym^2\phi)|_{\fh}$ is simple.\\
(c) The representation $(\Sym^2\phi)|_{\fh}$ is simple and $(\Alt^2\phi)|_{\fh}$ splits into two simple
components. No simple component occurs more than once.\\
(d) The representation $(\phi\otimes\phi)|_{\fh}$ splits into three different simple components.\\
(e) The vector space of all complex matrices commuting with $(\phi\otimes\phi)|_{\fh}$ has dimension
three.
\end{theorem}

\begin{proof}
Note that $\Sym^2\phi=\phi_{(2,0,\ldots,0)}$ is equivalent to the adjoint representation
and is always simple (see Lemma~\ref{adjoint} of Appendix~\ref{Dynkin_techniques}).
We have $\Alt^2\phi=\phi_{(0,1,0,\ldots,0)}\oplus\phi_{(0,\ldots,0)}$ which follows from the discussion on
pp.~259--262 of Ref.~\onlinecite{FH91}
or pp.~206--209 of Ref.~\onlinecite{Bourb08b}.
We obtain that (b) and (c) are a consequence of (a). Obviously, (b) follows from (c).
If $\fh$ is a proper subalgebra of $\usp(\ell)$, the adjoint representation $\Sym^2\phi$ is no longer
simple when restricted to $\fh$ (see Lemma~\ref{adjoint} of Appendix~\ref{Dynkin_techniques}). 
Thus, (b) implies (a). The statements (c) and (d) are equivalent as $(\Alt^2\phi)|_{\fh}$ splits into 
at least two components. The equivalence of (d) and (e) follows from Lemma~\ref{reprtheory} of Appendix~\ref{app_pre}.
\end{proof}

We emphasize that the proof applies the irreducibility of the adjoint representation of 
$\usp(\ell)$ for $\ell\geq 2$.---Combining Theorems~\ref{thm_alt_square} and \ref{thm_sym_square} we obtain a 
second proof
of the classification of representations for which both the alternating square and the symmetric square
are simple (see Ref.~\onlinecite{ZS11}).
\begin{theorem}[see Thm.~54 of Ref.~\onlinecite{ZS11}]\label{su_tensor_square}
Let $\phi$ denote a faithful representation of a compact semisimple Lie algebra
$\fg$ such that both the alternating square $\Alt^2\phi$
and the symmetric square $\Sym^2\phi$ are simple.
Then $\fg=\su(\ell{+}1)$ with $\ell\geq 1$ and $\phi$ is 
(up to outer automorphisms of $\fg$) the standard representation with highest weight
$(1,0,\ldots,0)$.
\end{theorem}

The classification culminates into a convenient necessary and sufficient
condition for deciding if a subalgebra of $\su(\ell{+}1)$ is proper (which also proves Theorem~\ref{ThmA}).
\begin{theorem}[see Thm.~21 of Ref.~\onlinecite{ZS11}]
Given a subalgebra $\fh$ of $\su(\ell{+}1)$ with $\ell\geq 1$
and the standard representation $\phi:=\phi_{(1,0,\ldots,0)}$ of $\su(\ell{+}1)$,
 the following statements are equivalent:\\
(a) $\fh=\su(\ell{+}1)$.\\
(b) The representations
$(\Alt^2\phi)|_{\fh}$ and
$(\Sym^2\phi)|_{\fh}$ are simple.\\
(c) The representations
$(\Alt^2\phi)|_{\fh}$ and
$(\Sym^2\phi)|_{\fh}$ are simple. No simple component occurs more than once.\\
(d) The representation $(\phi\otimes\phi)|_{\fh}$ splits into two different simple components.\\
(e) $\dim(\comm[(\phi\otimes\phi)|_{\fh}])=2$.
\end{theorem}

\begin{proof}
Statement (b) follows from (a) due to Theorem~\ref{su_tensor_square}.
We apply Lemma~\ref{plus}(iii) of Appendix~\ref{app_pre} to statement (b) and obtain that $\phi|_{\fh}$
is simple. Therefore, the commutant of $\phi(\fh)$ is
trivial, i.e., it is equal to complex multiples of the identity. We conclude that the centralizer 
of  $\phi(\fh)$ in $\phi[\su(\ell{+}1)]$ is zero and the center of $\fh$ is also zero.
Thus, $\fh$ is semisimple and we can use  Theorem~\ref{su_tensor_square}
to prove (a). Obviously, (b) and (c) are equivalent as $(\Alt^2\phi)|_{\fh}$
and $(\Sym^2\phi)|_{\fh}$ differ if both are simple.
The statements (c) and (d) follow from each other
as $(\phi\otimes\phi)|_{\fh}=(\Alt^2\phi)|_{\fh} \oplus (\Sym^2\phi)|_{\fh}$.
The equivalence of (d) and (e) is a consequence of Lemma~\ref{reprtheory} in Appendix~\ref{app_pre}.
\end{proof}

\section{Classification-free Results\label{classification-free}}
Building on the approach of Section~\ref{Alt_Sym_Ten}, we develop now
classification-free methods leading to general results for distinguishing
compact semisimple Lie algebras from its subalgebras. This
will in particular provide a proof for Theorem~\ref{ThmB}.
We start by introducing and discussing one- and two-``norms'' in 
Section~\ref{one_and_two}, and continue by 
relating  $\phi \otimes \psi$ to $\phi \otimes \bar{\psi}$
for semisimple representations $\phi$ and $\psi$ of a compact Lie algebra
(see Section~\ref{commutant_sect}).
In Section~\ref{use_adjoint}, we apply properties of the adjoint
representation in order to prove the central result of Theorem~\ref{one_norm_cond}.
We summarize our classification-free results in Section~\ref{class-free} by presenting
a set of statements which are equivalent to $\fh=\fg$ for a subalgebra $\fh$ of
a compact semisimple Lie algebra $\fg$.

\subsection{One- and two-``norms''\label{one_and_two}}
We consider a semisimple representation 
$\phi$ of a compact Lie algebra $\fg$.
The decomposition of $\phi$ into simple representations
can be given in the form $\oplus_{i\in \mathcal{I}} \left[ \unity_{m_i} \otimes \phi_i \right]$
of Lemma~\ref{reprtheory} in Appendix~\ref{app_pre}, where $m_i$ denotes the corresponding multiplicity.
In the following, we will use the more concise notation $\oplus_{i\in \mathcal{I}}\, \phi_i^{\oplus m_i}$.
It will be convenient to introduce the notations $\normone{\phi}:=\sum_{i\in\mathcal{I}} m_i$ 
and $\normtwo{\phi}:=\sum_{i\in\mathcal{I}} m_i^2$. We have as an immediate 
consequence of Lemma~\ref{reprtheory} in Appendix~\ref{app_pre}
that $\normtwo{\phi}$ is equal to the dimension of the commutant of $\phi(\fg)$.
Moreover, we obtain the following propositions.
\begin{proposition}\label{sumrules}
Consider two semisimple representations $\phi\cong \oplus_{i\in \mathcal{I}}\, \phi_i^{\oplus m_i}$
and  $\psi \cong \oplus_{i\in \mathcal{I}}\, \psi_i^{\oplus n_i}$
of a compact Lie algebra $\fg$ which decomposes into simple representations
$\phi_i$ and $\psi_i$ with multiplicities $m_i$ and $n_i$, respectively.
One obtains (i) $\normone{ \phi \oplus \psi} = \normone{\phi} + 
\normone{\psi}$ and (ii) $\normtwo{ \phi \oplus \psi} \geq \normtwo{\phi} + 
\normtwo{\psi}$.
\end{proposition}

\begin{proof}
Statement (i) follows from $\normone{ \phi \oplus \psi}= \sum_{i\in\mathcal{I}} m_i + n_i 
= \normone{ \phi} + \normone{\psi}$. 
Similarly, (ii) is a consequence of
$\normtwo{\phi \oplus \psi}= \sum_{i\in\mathcal{I}} (m_i + n_i)^2 = \sum_{i\in\mathcal{I}}
m_i^2 + 2 m_i n_i + n_i^2 \geq \sum_{i\in\mathcal{I}}
m_i^2 + n_i^2=\normtwo{\phi} + \normtwo{\psi}$.
\end{proof}

\begin{proposition}\label{one_two}
Consider a semisimple representation $\phi\cong \oplus_{i\in \mathcal{I}}\, \phi_i^{\oplus m_i}$
of a compact Lie algebra $\fg$ which decomposes into simple representations
$\phi_i$ with multiplicities $m_i$. The restrictions of $\phi$ and $\phi_i$ to a subalgebra $\fh$ of $\fg$
are given by $\phi|_{\fh}\cong \oplus_{j \in \mathcal{J}} \psi_j^{\oplus n_j}$
and $(\phi_i)|_{\fh}\cong \oplus_{j \in \mathcal{J}} \psi_j^{\oplus n_{ji}}$
where $\psi_j$ denotes a simple representation of $\fh$ and $n_j = \sum_{i\in \mathcal{I}} n_{ji}\, m_i$.
We obtain\\
(a) $\normone{(\phi_i)|_{\fh}}=1$ for all $i \in \mathcal{I}$ with $m_i\neq 0$ if and only if
$\normone{\phi|_{\fh}} = \normone{\phi}$,\\
(b) $\normtwo{\phi|_{\fh}} = \normtwo{\phi}$ implies $\normone{\phi|_{\fh}} = \normone{\phi}$,\\
(c) $\normtwo{\phi|_{\fh}} = \normtwo{\phi}$  
$\Leftrightarrow$
$\normone{\phi|_{\fh}} = \normone{\phi}$ and $(\phi_i)|_{\fh} \not\cong (\phi_k)|_{\fh}$ holds 
for all  $i,k \in \mathcal{I}$ with $i\ne k$, $m_i\neq 0$, and $m_k\neq 0$.
\end{proposition}

\begin{proof}
We assume during the proof that $\mathcal{I}$ contains only elements $i$ with $m_i\neq 0$. 
Note that $\normone{\phi}=\sum_{i \in \mathcal{I}} m_i$, $\normone{(\phi_i)|_{\fh}}=
\sum_{j \in \mathcal{J}} n_{ji}$, 
and $\normone{\phi|_{\fh}}=\sum_{j \in \mathcal{J}} n_{j}=\sum_{j \in \mathcal{J}} 
( \sum_{i \in \mathcal{I}} n_{ji}\, m_i)=\sum_{i \in \mathcal{I}} ( \sum_{j \in \mathcal{J}} n_{ji})\, m_i$. 
Furthermore, $\normone{\phi|_{\fh}} = \normone{\phi}$ if and only if $\sum_{j \in \mathcal{J}} n_{ji}=1$ for all
$i \in \mathcal{I}$ if and only if $\normone{(\phi_i)|_{\fh}}=1$ for all $i \in \mathcal{I}$. This completes the proof of (a).
We remark that $\normtwo{\phi}=\sum_{i \in \mathcal{I}} m_i^2$ and use the multinomial theorem to obtain
\begin{align*}
\normtwo{\phi|_{\fh}}&=\sum_{j \in \mathcal{J}} n_{j}^2=\sum_{j \in \mathcal{J}} ( \sum_{i \in \mathcal{I}} n_{ji}\, m_i)^2=
\sum_{j\in\mathcal{J}} 
(2\sum_{\genfrac{}{}{0pt}{}{i,\ell\in \mathcal{I}}{i\ne \ell}} n_{ji}\, m_i\, n_{j\ell}\, m_\ell
+\sum_{i\in\mathcal{I}} n_{ji}^2\, m_i^2)\\
&=
2\sum_{\genfrac{}{}{0pt}{}{i,\ell\in \mathcal{I}}{i\ne \ell}} 
(\sum_{j\in\mathcal{J}} 
n_{ji}\, n_{j\ell})\,
m_i \, m_\ell
+\sum_{i\in\mathcal{I}}
(\sum_{j\in\mathcal{J}} n_{ji}^2)\, m_i^2.
\end{align*}
Note that $\sum_{j\in\mathcal{J}} n_{ji}^2\geq 1$ for all $i \in \mathcal{I}$. 
We get from $\normtwo{\phi|_{\fh}} = \normtwo{\phi}$ that (i)
$\sum_{j\in\mathcal{J}} n_{ji}^2=1$ for all $i \in \mathcal{I}$
and that (ii)
$\sum_{j\in\mathcal{J}} 
n_{ji}\, n_{j\ell}=0$ for all $i,\ell \in \mathcal{I}$ with $i\ne \ell$. Condition (i) implies
that $\sum_{j\in\mathcal{J}} n_{ji}=1$ holds  for all $i \in \mathcal{I}$.
We can now prove (b)
by applying (a) to the fact that $\normone{(\phi_i)|_{\fh}}=\sum_{j \in \mathcal{J}} n_{ji}=1$
is valid for all $i \in \mathcal{I}$.
We consider now the statement (c). The fact that the condition 
($*$) $(\phi_i)|_{\fh} \not\cong (\phi_\ell)|_{\fh}$ holds for all  $i,\ell \in \mathcal{I}$ with $i\ne \ell$
is implied by
(i) and (ii). This completes the direction ``$\Rightarrow$''. It follows (i) from $\normone{\phi|_{\fh}} = \normone{\phi}$
by applying (a). The conditions (i) and ($*$) imply (ii) and
the conditions (i) and (ii) imply $\normtwo{\phi|_{\fh}} = \normtwo{\phi}$. This completes 
the direction
``$\Leftarrow$''.
\end{proof}
Note that the converse of part (b) in Proposition~\ref{one_two} is in general false.

\subsection{From $\phi \otimes \psi$ to $\phi \otimes \bar{\psi}$\label{commutant_sect}}
Here, we provide a Lie-algebraic argument why $\normtwo{\phi \otimes \psi}=\normtwo{\phi \otimes \bar{\psi}}$
holds for semisimple representations $\phi$ and $\psi$ of a compact Lie algebra.

\begin{proposition}\label{commutant_conjugated}
Given a compact Lie algebra $\fg$ and two  semisimple  
representations $\phi$ and $\psi$ of $\fg$, it follows that
$\comm (\phi \otimes \psi) = \unity \otimes \tau [\comm(\phi \otimes \bar{\psi})]$,
where $\tau$ denotes the transpose operation (acting on the second 
tensor component).
\end{proposition}
\begin{proof}
Let us recall that the dual representation is given
by  $\bar{\phi}(g):=-\phi(g)^T$ for $g \in \fg$, hence
$\phi \otimes \bar{\psi} (g)= \phi(g) \otimes \unity - \unity \otimes \psi(g)^T$. 
Suppose that $\sum_{i} v_i \otimes w_i \in \comm (\phi \otimes \psi)$, i.e.,
for any $g \in \fg$ one has $[(\phi \otimes \psi)(g), \sum_i v_i \otimes w_i]=0$. 
Considering the commutator of the element $\unity \otimes \tau (\sum_{i} v_i \otimes w_i ) 
= \sum_{i} v_i \otimes w_i^T $ with $(\phi\otimes \bar{\psi})(g)$, we arrive at the relation 
$[\sum_{i} v_i \otimes w_i^T , (\phi\otimes \bar{\psi})(g)]=\sum_{i}[v_i, \phi(g)] \otimes w_i^T - 
\sum_i v_i \otimes [w_i^T, \psi(g)^T]=\sum_{i}[v_i, \phi(g)] \otimes w_i^T + \sum_i v_i \otimes 
[w_i, \psi(g)]^T= \unity \otimes \tau([\sum_{i} v_i \otimes w_i , (\phi\otimes \psi)(g))] = \unity \otimes 
\tau (0)=0$, thus $(\unity \otimes \tau) (\sum_{i} v_i \otimes w_i) \in \comm(\phi\otimes \bar{\psi})$.
Completely analogously one can prove that for any $\sum p_i \otimes q_i \in \comm(\phi\otimes \bar{\psi})$ 
one has that $\unity \otimes \tau^{-1} (\sum_i  p_i \otimes q_i) = \unity \otimes \tau (\sum_i  p_i \otimes q_i) 
\in \comm(\phi\otimes \psi)$. 
This completes the proof.
\end{proof}

Proposition~\ref{commutant_conjugated} can be readily applied in the proof
of the following proposition.

\begin{proposition}\label{two_norm_conjugated}
Given a compact Lie algebra $\fg$ and two  semisimple 
representations $\phi$ and $\psi$ of $\fg$, it follows that
$\normtwo{\phi \otimes \psi}=\normtwo{\phi \otimes \bar{\psi}}$.
\end{proposition}
\begin{proof}
According to Lemma~\ref{reprtheory} of Appendix~\ref{app_pre}, it follows that 
$\normtwo{\phi \otimes \psi}=\dim \comm(\phi \otimes \psi)$ 
and $\normtwo{\phi \otimes \bar{\psi}}=\dim \comm(\phi \otimes \bar{\psi})$. 
From Proposition~\ref{commutant_conjugated} we know that 
$\comm(\phi \otimes \psi) $ is mapped by a non-degenerate linear map (the partial transpose) 
to $\comm (\phi \otimes \bar{\psi})$, so the dimensions 
of the two commutants are equal, thus also 
$\normtwo{\phi \otimes \psi}$ and $\normtwo{\phi \otimes \bar{\psi}}$ are equal.
\end{proof}

\subsection{Using the adjoint representation\label{use_adjoint}}
The adjoint representation plays a important part in our argument, and we 
recall and develop now some of its properties in order to prove
our central result of Theorem~\ref{one_norm_cond} as given below.

\begin{proposition}\label{adjoint_omnibus}
Consider a compact semisimple Lie algebra $\fg$ and its decomposition 
$\fg \cong \oplus_{i\in\mathcal{I}} \fg_i$ into  simple ideals $\fg_i$. (a) The adjoint representation 
$\theta_{\fg}$ of $\fg$ decomposes as
$\theta_{\fg}\cong \oplus_{i\in\mathcal{I}} \theta_{\fg_i}$. (b) It is 
simple if $\fg$ is simple.
(c) The adjoint representation $(\theta_{\fg})|_{\fh}$ of $\fg$ restricted to a proper subalgebra $\fh$
is reducible. (d) The adjoint representation  $\theta_{\fh}$ of $\fh$ occurs as a subrepresentation
of $(\theta_{\fg})|_{\fh}$.
\end{proposition}
\begin{proof}
The statement (a) is apparent.
The statements (b) and (c) follow from Lemma~\ref{adjoint} of Appendix~\ref{Dynkin_techniques}.
The adjoint representation $\theta_{\fg}$ of $\fg$  constitutes an action $\fg \times \fg \to \fg$
which is defined using the commutator $[g_1,g_2]=g_3$ for $g_i$ in $\fg$. If one restricts $\theta_{\fg}$ 
to elements of $\fh$, the representation $(\theta_{\fg})|_{\fh}$ forms an action 
$\fh \times \fg \to \fg$ by $[h,g_2]=g_3$ with $h\in \fh$ and $g_i$ in $\fg$. 
The adjoint representation $\theta_{\fh}$ occurs as a subrepresentation as
$[h_1,h_2] \in \fh$ for $h_i \in \fh$. Statement (d) follows.
\end{proof}

This immediately implies the following result.

\begin{proposition}\label{adjoint_one_norm_cond}
Given a subalgebra $\fh$ of a compact semisimple Lie algebra $\fg$ and
their adjoint representations $\theta_{\fh}$ and $\theta_{\fg}$, then
$\fh = \fg$ if and only if $\normone{(\theta_{\fg})|_{\fh}}=\normone{\theta_{\fg}}$.
\end{proposition}
\begin{proof}
One has to show that $\normone{(\theta_{\fg})|_{\fh}}\neq \normone{\theta_{\fg}}$
if $\fh \neq \fg$. Assuming that $\fh \neq \fg$, Proposition~\ref{adjoint_omnibus}(c) implies that 
$(\theta_{\fg})|_{\fh}$ has more simple 
components than $\theta_{\fg}$. It follows that $\normone{(\theta_{\fg})|_{\fh}}> \normone{\theta_{\fg}}$
which concludes the proof.
\end{proof}

Let us recall some well-known connection between the standard and the adjoint representation
of $\su(\ell+1)$.

\begin{proposition}\label{adjoint_su}
Given the standard representation $\kappa$ of $\su(\ell+1)$, one obtains
$\kappa \otimes \bar{\kappa} \cong 1 \oplus \theta_{\su(\ell+1)}$, where 
$1$ denotes the trivial representation and 
$\theta_{\su(\ell+1)}$ denotes the adjoint representation of $\su(\ell+1)$.
\end{proposition}
\begin{proof}
For $\su(\ell+1)$, the standard representation $\kappa$, the dual $\bar{\kappa}$ of $\kappa$,
the adjoint representation $\theta_{\su(\ell+1)}$, and the trivial representation have
highest weights  $(1,0,\ldots,0)$, $(0,\ldots,0,1)$, $(1,0,\ldots,0,1)$ (and $(2)$ for $\ell=1$),
and $(0,\ldots,0)$, respectively.
The proposition can now 
be inferred from the statements on p.~225 of Ref.~\onlinecite{FH91}.
\end{proof}

We can now combine all previous results in this section in order to prove 
the following central theorem.

\begin{theorem}\label{one_norm_cond}
Given a subalgebra $\fh$ of a compact semisimple Lie algebra $\fg$
and a faithful representation $\phi$ of $\fg$,
then $\fh = \fg$ if and only if $\normone{(\phi \otimes \bar{\phi})|_{\fh}}=\normone{\phi \otimes \bar{\phi}}$.
\end{theorem}
\begin{proof}
Let $d:=\dim(\phi)$. As $\phi$ is faithful, $\fg\subseteq \su(d)$. Given the standard reprepresentation
$\kappa$ of $\su(d)$, we obtain $\phi=\kappa|_{\fg}$ 
and $\phi \otimes \bar{\phi}=\kappa|_{\fg} \otimes \bar{\kappa}|_{\fg}=(\kappa \otimes \bar{\kappa})|_{\fg}$.
Using Propositions~\ref{adjoint_su} and \ref{adjoint_omnibus}(d) it follows 
that $\theta_{\su(d)}$ occurs in $\kappa \otimes \bar{\kappa}$ 
and $\theta_{\fg}$ occurs in $\theta_{\su(d)}|_{\fg}$. Therefore, $\theta_{\fg}$ occurs in $\phi \otimes \bar{\phi}$.
But $\theta_{\fg}$ splits when restricted to $\fh\ne \fg$ [see Proposition~\ref{adjoint_omnibus}(c)]
and it follows that $\normone{(\theta_{\fg})|_{\fh}}>
\normone{\theta_{\fg}}$
[see Proposition~\ref{adjoint_one_norm_cond}]. 
We apply Proposition~\ref{one_two}(a)
and conclude that there exists a simple representation in the decomposition
of $\theta_{\fg}$ which splits when restricted to $\fh$. But this simple representation
also appears in the decomposition of $\phi \otimes \bar{\phi}$, and we can apply 
Proposition~\ref{one_two}(a) again to conclude that  $\normone{(\phi \otimes \bar{\phi})|_{\fh}}>
 \normone{\phi \otimes \bar{\phi}}$.
\end{proof}

\subsection{Classification-free theorem\label{class-free}}

Recall that $\comm[\phi]$ denotes the commutant of a representation $\phi$ of a Lie algebra $\fg$.
We summarize our results in a convenient omnibus theorem which relates
the equality of a compact semisimple Lie algebra $\fg$ to
one of its subalgebras $\fh$
to the condition $\dim(\comm[(\phi \otimes \phi)|_{\fh}]) = \dim(\comm[\phi \otimes \phi])$ for 
a faithful representation $\phi$ of $\fg$
(and thereby also
proving Theorem~\ref{ThmB}), as well as to various other variants.
\begin{theorem} \label{thm:tensor_square}
Given a subalgebra $\fh$ of a compact semisimple Lie algebra $\fg$
and a faithful representation $\phi$ of $\fg$,
then the following statements are equivalent:\\
(1) $\fh = \fg$,\\
(2) $\dim(\comm[(\phi \otimes \bar{\phi})|_{\fh}]) = \dim(\comm[\phi \otimes \bar{\phi}])$,\\
(3) $\dim(\comm[(\phi \otimes \phi)|_{\fh}]) = \dim(\comm[\phi \otimes \phi])$,\\
(4)  $\normtwo{(\phi \otimes \bar{\phi})|_{\fh}}=\normtwo{\phi \otimes \bar{\phi}}$,\\
(5)  $\normone{(\phi \otimes \bar{\phi})|_{\fh}}=\normone{\phi \otimes \bar{\phi}}$,\\
(6)  $\normtwo{(\phi \otimes \phi)|_{\fh}}=\normtwo{\phi \otimes \phi}$.\\
(7)  $\normone{(\phi \otimes \phi)|_{\fh}}=\normone{\phi \otimes \phi}$.
\end{theorem}

Note that the proof of the equivalence of (7) relies on 
a theorem of Coquereaux and Zuber (see Appendix~\ref{CZ}) for which a classification-free proof is 
(to our knowledge)
not yet known.
\begin{proof}
Obviously, statement (1) implies all the other ones. The equivalence of (2) and (4) as well as (3) and (6)
is a consequence of Lemma~\ref{reprtheory} in Appendix~\ref{app_pre}. 
One obtains the equivalence of (4) and (6) by applying Proposition~\ref{two_norm_conjugated}.
Proposition~\ref{one_two}(b) shows that (5) is implied by (4). It remains to prove that
(5) implies (1) which follows by Theorem~\ref{one_norm_cond}.
Note that (5) is equivalent to (7) by applying Proposition~\ref{one_norm_conjugated}
from Appendix~\ref{CZ}.
\end{proof}

\section{Discussion\label{discussion}}
It is natural to ask if and how one could extend the results of Theorem~\ref{thm:tensor_square}
to general compact Lie algebras.
A first step in this direction is given by the following theorem 
which
allows us to conclude that the semisimple parts of $\fh$ and $\fg$ have to be equal if
the dimensions of the commutant of the tensor squares are equal assuming that 
the corresponding representation is both faithful and semisimple.
\begin{theorem}\label{Thm_semsimple}
Consider a subalgebra $\fh$ of a compact Lie algebra $\fg= \fs(\fg) \oplus \fc(\fg)$ 
(which decomposes into its semisimple part  $\fs(\fg)$ and its center $\fc(\fg)$)
as well as a faithful and semisimple representation $\phi$ of $\fg$. 
One has $\fs(\fh)=\fs(\fg)$ if 
$\dim(\comm[(\phi \otimes \phi)|_{\fh}]) = \dim(\comm[\phi \otimes \phi])$.
\end{theorem}
The corresponding proof is given in Appendix~\ref{app_semisimple}.
It is important to emphasize that the converse of Theorem~\ref{Thm_semsimple}
is in general not true as the following counter-example shows: 
Consider the two Lie-algebra generators $A$ and $B$ which are given
in a faithful representation $\phi$ as commuting matrices 
$$
\phi(A):=
\begin{pmatrix}
1 & 0 & 0\\
0 & 1 & 0\\
0 & 0 & -2
\end{pmatrix}
\;\text{ and }\;
\phi(B):=
\begin{pmatrix}
2 & 0 & 0\\
0 & -1 & 0\\
0 & 0 & -1
\end{pmatrix}.
$$
Note that $A$ generates a one-dimensional abelian Lie algebra $\fh$ isomorphic to $\uu(1)$
and that $A$ and $B$ generate a two-dimensional abelian Lie algebra $\fg$ isomorphic to $\uu(1)\oplus\uu(1)$.
Their semisimple parts $\fs(\fh)=\fs(\fg)=\{0\}$ are equal and trivial, but 
$\dim(\comm[(\phi \otimes \phi)|_{\fh}]) = 33$ and $\dim(\comm[(\phi \otimes \phi)|_{\fg}])=15$.
In order to treat the case of a general compact Lie algebra completely, an
additional condition for guaranteeing the equality of the centers of
$\fh$ and $\fg$ will be necessary. This is the topic 
of a related study\cite{ZZSB15a} which focuses
on control-theoretic applications in quantum systems.

A different possibility of extending our results 
in Theorem~\ref{thm:tensor_square} is by
providing not only inequalities  as 
$ \normone{(\phi \otimes \phi)|_{\fh}}  > \normone{\phi \otimes \phi}$ and
$ \normtwo{(\phi \otimes \phi)|_{\fh}}  > \normtwo{\phi \otimes \phi}$.
But one can in special cases determine bounds on the gap 
in these inequalities.
If $\fg$ is simple and $\phi$ is
self-dual, one can apply a result of  King and Wybourne\cite{KW96}
in order to derive the following bounds (see Appendix~\ref{sharp_app}).

\begin{proposition}\label{bound_self_dual}
Let $\alpha$ be a simple and  self-dual representation  of a compact simple
Lie algebra $\fg$, and let $\fh$ be a 
 subalgebra of $\fg$, then \\
(1)  $\normone{(\alpha \otimes \alpha)|_{\fh}} \ge b(\alpha) + \normone{\alpha \otimes \alpha}$,\\
(2)  $\normtwo{(\alpha \otimes \alpha)|_{\fh}} \ge b(\alpha)^2 + 
\normtwo{\alpha \otimes \alpha}$, and\\
(3) $\dim(\comm[(\alpha \otimes \alpha)|_{\fh}]) \ge b(\alpha)^2 +\dim(\comm[\alpha \otimes \alpha])$ hold,\\
where $b(\alpha)$ denotes the number of non-vanishing components in the highest weight $(\alpha_1,\allowbreak \ldots, \allowbreak \alpha_\ell)$
corresponding to $\alpha$.
\end{proposition}

Using the structures of $\fg$ and its representation $\phi$,
it would be interesting to derive bounds which are 
better or which are applicable in other cases.

\begin{acknowledgments}
R.Z. acknowledges support from the {\em Deutsche Forschungsgemeinschaft}
({\sc dfg}) in the collaborative research center {\sc sfb} 631 as well as from
the EU programmes {\sc siqs} and {\sc quaint}.
Z.Z. acknowledges support from
the British {\em Engineering and Physical Sciences Research Council} ({\sc epsrc}).
\end{acknowledgments}

\appendix

\section{Preliminaries\label{app_pre}}

In this appendix, we recall some basic facts about compact Lie algebras
and their representations. Connections to the complexifications
are highlighted in Appendix~\ref{compact_Lie}. This is followed by the well-known classification of simple
representations into symplectic, orthogonal, and unitary ones. 
The definition and discussion of alternating and symmetric squares are given in
Appendix~\ref{def_alt_sym}.

\subsection{Compact Lie Algebras and their Representations\label{compact_Lie}}

Any compact (real) Lie algebra can be written as a direct sum of two compact Lie algebras
where one is semisimple and the other one is commutative (see Chap.~IX, Sec.~1.3, Prop.~1 of 
Ref.~\onlinecite{Bourb08b} and Chap.~I, Sec.~6.4, Prop.~5 of 
Ref~\onlinecite{Bourb89}). The semisimple Lie algebra can be further decomposed into
a direct sum of compact, simple Lie algebras which
consists of the classical ones $\su(\ell+1)$, $\so(2\ell+1)$, $\usp(\ell):=\usp(2\ell,\C) \cap \su(2\ell)$,
and $\so(2\ell)$ where $\ell \in \N\setminus\{0\}$ 
(excluding $\so(2)$ and $\so(4) \cong \su(2) \oplus \su(2)$ which are not simple) as well as the exceptional ones
$\fg_2$, $\ff_4$, $\fe_6$, $\fe_7$, $\fe_8$ (see Ref.~\onlinecite{Bourb08b}). 
Note the isomorphisms $\su(2) \cong \so(3) \cong \usp(1)$, $\so(5) \cong \usp(2)$, and 
$\su(4) \cong \so(6)$ (see, e.g., Thm.~X.3.12 in Ref.~\onlinecite{Helgason78}).  

We briefly recall some elementary facts about the complexification of a real Lie algebra
which will be directly applied to compact Lie algebras.
Let $\phi$ denote a (finite-dimensional) complex matrix representation
of a real Lie algebra $\fk$, where $\phi$ maps a Lie-algebra element $k \in \fk$ to a 
square matrix $\phi(k) \in \gl(d,\C)$ of finite degree $d:=\dim(\phi)$,
where $\gl(d,\C)$ denotes
the set of complex $d\times d$-matrices. All representations
in this manuscript are matrix representations with complex matrix entries.
For every real Lie algebra $\fk$, a representation
$\phi$ of $\fk$ in $\gl(d,\C)$ can be naturally extended to 
a representation $\phi_\C$ of $\fk_\C$ in $\gl(d,\C)$, where 
$\fk_{\C}:=\fk + i \fk$ is the complexification of $\fk$.

\begin{lemma}\label{complex}
Consider a real Lie algebra $\fk$
and its complexification $\fk_{\C}$. 
Let $\tau$ denote a representation of $\fk_\C$. Moreover, let $\phi$ and $\psi$ denote two 
representations of $\fk$, and 
$\phi_\C$ and $\psi_\C$ the corresponding representations of $\fk_\C$. The following statements hold:\\
(i) The representations $\phi$ and $\psi$ are isomorphic if and only if the representations
$\phi_{\C}$ and $\psi_{\C}$ are isomorphic.\\
(ii) For the representation $\tau$ of $\fk_{\C}$, there exists a representation $\eta$ of $\fk$ such that  $\tau=\eta_{\C}$.\\
(iii) $\phi$ is simple (i.e., irreducible) if and only if $\phi_{\C}$ is simple.\\
(iv) $\phi$ is semisimple (i.e., completely reducible) if and only if $\phi_{\C}$ is semisimple.\\
(v) If two real Lie algebras are isomorphic, then so are their complexifications.\\
(vi) $\fk$ is abelian if and only if $\fk_{\C}$ is abelian.\\
(vii) $\fk$ is simple if $\fk_{\C}$ is simple.\\
(viii) $\fk$ is semisimple if and only if $\fk_{\C}$ is semisimple.\\
(ix) $\fk$ is reductive if and only if $\fk_{\C}$ is reductive.
\end{lemma}

\begin{proof} 
Part (i) is trivial. In order to prove (ii), we set $\eta:=\tau |_{\fk}$. The $\R$-linearity 
of $\eta$ follows from the $\C$-linearity of $\tau$. We obtain
$\eta_{\C}(\fk_{\C})=c_1 \eta(\fk) + c_2 \eta(\fk)= c_1 (\tau |_{\fk})(\fk) + c_2 (\tau |_{\fk})(\fk)
=c_1 \tau(\fk) + c_2 \tau(\fk)  = \tau(c_1 \fk+ c_2\fk)= \tau(\fk+ i\fk) = \tau(\fk_{\C})$ with $c_j \in \C$,
and part (ii) follows. We address now part (iii): If $\phi$ is simple, so is $\phi_{\C}$. Assume that $\phi_{\C}$ 
is simple and $\phi$ is not simple. Thus, there exists two linear-independent complex matrices $M_1$ and $M_2$
which commute with $\phi(\fk)$, $i\phi(\fk)$, and $\phi(\fk)+i\phi(\fk)=\phi_{\C}(\fk_{\C})$. This proves
(iii) by contradiction. Given that $\phi$ is semisimple, we obtain that $\phi=\oplus_{k}\phi_k$ with $\phi_k$ simple.
It follows that $\phi_{\C}$ is semisimple as it is isomorphic to $\oplus_{k} (\phi_k)_{\C}$. Assuming that
$\phi_{\C}$ is semisimple, we have $\phi_{\C}=\oplus_k \tau_k$ where $\tau_k$ are
simple representations of $\fk_{\C}$. Applying (ii) we obtain (complex) representations $\eta_k$ with
$\tau_k=(\eta_k)_{\C}$. The representation $\phi$ is semisimple as it is isomorphic to $\oplus_k \eta_k$,
which completes the proof of (iv). Assume that the real Lie algebras $\fk$ and $\fk'$ 
are isomorphic. It follows that $\fk_{\C}=\fk+i\fk\cong\fk'+i\fk'=\fk'_{\C}$ which proves (v).
For (vi)-(ix), we refer to Chap.~I, Sec.~1.9 and Sec.~6.10 of Ref.~\onlinecite{Bourb89}.
\end{proof}

We can sharpen Lemma~\ref{complex} if we assume that the considered real Lie algebra is compact.
Two compact Lie algebras $\fg$ and $\fg'$ are isomorphic iff their complexifications $\fg_{\C}$ and $\fg'_{\C}$
are isomorphic (see Cor.~1 of Thm.~1 in Chap.~IX, Sec.~3.3 of Ref.~\onlinecite{Bourb08b}). 
Referring to p.~283 and p.~300 of Ref.~\onlinecite{Bourb08b}, we also obtain that
a compact Lie algebra $\fg$ and its complexification $\fg_{\C}$ are reductive 
as well as that $\fg$ is simple (resp.\ semisimple) if and only if $\fg_{\C}$ is simple (resp.\ semisimple).
\begin{lemma}\label{complex_lemma}
Consider two compact Lie algebras $\fg$ and $\fg'$ 
as well as their complexifications $\fg_{\C}$ and $(\fg')_{\C}$. 
We obtain that
(i) $\fg$ is isomorphic to $\fg'$ iff
$\fg_{\C}$ is isomorphic to $(\fg')_{\C}$
as well as (ii) $\fg$ is simple iff $\fg_{\C}$ is simple.
\end{lemma}
As a consequence of the previous discussion, certain questions  about a compact Lie algebra 
$\fg$ (and its representations) will be answered by switching to
its complexification $\fg_{\C}$ without further comment.---Although a compact Lie algebra might
have a  representation which is not semisimple (e.g., consider a one-dimensional
Lie algebra with a representation which maps
a Lie-algebra element $\lambda$ to the matrix
$
\left(
\begin{smallmatrix}
0 & 0 \\
\lambda & 0
\end{smallmatrix}
\right)
$), any restriction of a semisimple
representation to one of its ideals or subalgebras stays semisimple.
\begin{lemma}\label{restr_semisimple}
Let us consider a semisimple (complex) representation $\phi$ of a compact Lie algebra $\fg$.
The restriction $\phi|_{\fh}$ of $\phi$ to any subalgebra $\fh$ of $\fg$ is still semisimple.
\end{lemma}

\begin{proof}
Note that the corresponding conclusion  for ideals 
holds for arbitrary  Lie algebras of characteristic zero, see
Chap.~I, Sec.~6.5, Cor.~4 to Thm.~4 of Ref.~\onlinecite{Bourb89}. In order to prove the statement,
we apply that both $\fg$ and its subalgebra $\fh$ are compact and consequently reductive. We have
$\fh=\fh_s \oplus \fh_a$ where $\fh_s$ is semsimple and $\fh_a$ is abelian. 
The representation $\phi|_\fh$ is semsimple if and only if the set $\phi(\fh_a)$ contains only
semisimple matrices
(i.e., all elements of $\phi(\fh_a)$ are diagonalizable over $\C$), 
see Chap.~I, Sec.~6.5, Thm.~4 of Ref.~\onlinecite{Bourb89}. 
We remark that $\fh_a$ is contained 
in a Cartan subalgebra $\ft_{\fg}$ of the compact Lie algebra $\fg$ as $\ft_{\fg}$ is a maximal 
abelian subalgebra (see  Chap.~IX, Sec.~2.1, Thm.~1 of Ref.~\onlinecite{Bourb08b}). 
In particular, $\ft_{\fg}$ is abelian and $\phi(\ft_{\fg})$ contains only
semisimple matrices if $\phi$ is semisimple and $\fg$ is reductive 
(see  Chap.~VII, Sec.~2.4, Cor.~3 to Thm.~2 of Ref.~\onlinecite{Bourb08b}).
Thus, 
$\phi(\fh_a)$ also contains only semisimple matrices which completes the proof.
\end{proof}

We also state some elementary consequences for the representations of compact Lie
algebras from the representation theory of compact groups (see Theorem~1.5 of Ref.~\onlinecite{Ledermann}).

\begin{lemma}\label{reprtheory}
Let $\phi$ denote a  semisimple (complex) representation 
of a compact Lie algebra $\fg$. The degree of $\phi$ is given by $d:=\dim(\phi)$.
The commutant algebra $\comm(\phi)$ is defined as all complex $d \times d$-matrices which
simultaneously
commute with $\phi(g)$ for all $g\in \fg$.
We obtain that
$\phi$ is equivalent to
$\oplus_{i\in \mathcal{I}} \left[ \unity_{m_i} \otimes \phi_i \right]$
where the distinct inequivalent simple representations $\phi_i$ of $\fg$ have 
degree $d_i$ and occur with multiplicity $m_i$ in the decomposition of  $\phi$.
In particular,
(a) $\dim \comm(\phi)=\sum_{i\in\mathcal{I}} m_i^2$,
(b) $\dim \centre(\comm(\phi))=\abs{\mathcal{I}}=\sum_{i\in \mathcal{I}, m_i \neq 0} 1$, and
(c) $d= \sum_{i\in\mathcal{I}} d_i m_i$.
\end{lemma}

Recall that 
a  bilinear form $B$ on 
$\C^{\dim(\phi)}$  is denoted as $\phi(\fg)$-invariant iff $B(\phi(g) v_1,v_2)=\allowbreak -B(v_1, \phi(g) v_2)$ for 
all $g\in\fg$ and $v_i \in \C^{\dim(\phi)}$. Let $\mathcal{B}_\phi$ denote the vector space 
of $\phi(\fg)$-invariant bilinear forms.
The \emph{dual representation} $\bar{\phi}$ is given
as  $\bar{\phi}(g):=-\phi(g)^T$ where $g \in \fg$ (see Chap.~I, Sec.~3.3
of Ref~\onlinecite{Bourb89}); note that $\bar{\phi}(g)=-\phi(g)^T=\overline{\phi(g)}$ if
$\phi(g)$ is skew-hermitian. We call a representation  \emph{self-dual}
if it is equivalent to its dual representation. Assuming that $\phi$ is simple and self-dual, 
it follows that $\dim(\mathcal{B}_\phi)=1$ and every nonzero element of $\mathcal{B}_\phi$ is nondegenerate
(see Chap.~VIII, Sec.~7.5 and App.~II.2 of Ref.~\onlinecite{Bourb08b} or 
Lemma~13 in Ref.~\onlinecite{ZS11}). 
A self-dual, simple representation $\phi$ is called either
\emph{orthogonal} or \emph{symplectic} depending on whether the corresponding
$\phi(\fg)$-invariant, nondegenerate bilinear form is symmetric or alternating (i.e., skew-symmetric), respectively.
A simple representation $\phi$ of $\fg$ is orthogonal [or symplectic]
iff $\phi(\fg)$ conjugate to a subalgebra of $\so(\dim(\phi))$ [or $\usp(\dim(\phi)/2)$]
(cf.\ p.~336 of Ref.~\onlinecite{Dynkin57} or Thm.~H on p.~144 of Ref.~\onlinecite{Samelson90}).
If $\phi$ is simple but not self-dual, then no nondegenerate bilinear form is $\phi(\fg)$-invariant 
and we denote $\phi$ as \emph{unitary}. The orthogonal, symplectic, and
unitary representations are also referred to as representations of real, quaternionic, or complex type, respectively.

A simple representation of a compact semisimple Lie algebra $\fg$ can be further
characterized using the theory of the highest weight:\cite{Bourb08b}
the highest weight is an integer vector $(x_1,\ldots, x_{\ell})$ with $x_j\geq 0$
and identifies classes of simple representations where the coefficient
$x_j$ indicates the contribution of the $j$-th fundamental weight (see Chap.~VIII, Sec.~7.2 of Ref.~\onlinecite{Bourb08b}). The coefficients
are ordered according to a convention of Bourbaki
(see Plate~I--IX of Ref.~\onlinecite{Bourb08a} or Table~\ref{dyn_tab}).
The length $\ell$ is equal to the rank of $\fg$ which is defined
as the dimension of any maximal abelian subalgebra of $\fg$.
Using the notation of the highest weight, simple representations of compact, simple Lie algebras
can be explicitly classified as symplectic, orthogonal, or unitary by the work of
Malcev~\cite{Malcev50} (see Refs.~\onlinecite{Dynkin57,BP70,McKay81,Samelson90}):

{
\fontsize{10}{12}
\renewcommand{\baselinestretch}{1}
\selectfont
\begin{table}[tb]
\caption{\label{dyn_tab}Dynkin diagrams for the compact simple Lie algebras  [$\so(2\ell)$ is only simple for $\ell\geq 3$]}
\begin{ruledtabular}
\begin{tabular}{c}
\includegraphics[scale=.825]{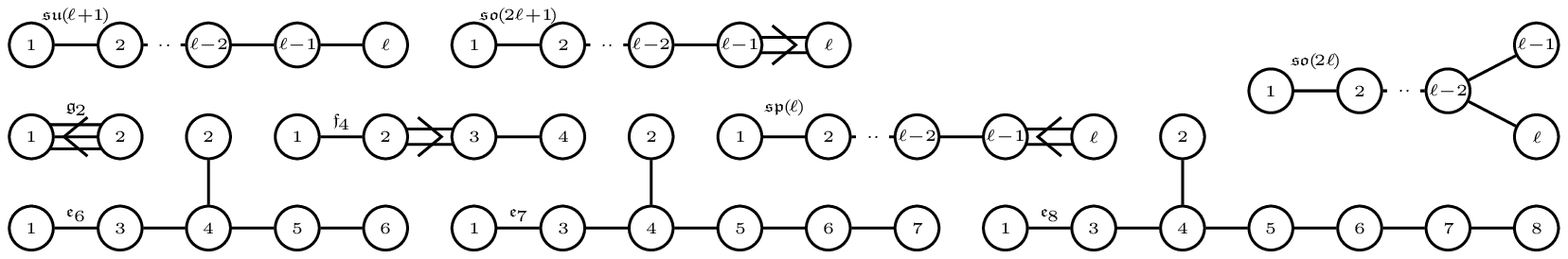}
\end{tabular}
\end{ruledtabular}
\end{table}
}

\begin{proposition}[Malcev]\label{Malcev}
A simple representation $\phi$ of a compact simple Lie algebra $\fg$
is either symplectic, orthogonal, or unitary depending on 
its highest weight $x=(x_1,\ldots,x_{\ell})$:
(1) Assuming that $\fg\cong\su(\ell+1)$, $\phi$ is unitary iff $x$ is not symmetric and 
symplectic iff $x$ is symmetric, $(\ell\bmod{4}) = 1$,
and $x_{((\ell-1)\, \mathrm{div}\, 2)+1}$ is odd; it is orthogonal
in all other cases. (2) In the case of $\fg\cong\so(2\ell+1)$ with $\ell \geq 3$,
we have that $\phi$ is symplectic iff $(\ell\bmod{4}) \in \{1,2\}$ and $x_\ell$ is odd,
while it is orthogonal otherwise. 
(3) For $\fg=\usp(\ell)$ with $\ell \geq 2$,
$\phi$ is symplectic if $\sum_{1\leq 2j+1 \leq \ell} x_{2j+1}$ is odd for $j\in\N\cup\{0\}$, 
while it is orthogonal otherwise.
(4) Assuming that $\fg=\so(2\ell)$ for $\ell \geq 3$, 
(a) $\phi$ is symplectic if $(\ell\bmod{4})=2$ and $x_{\ell-1}+x_{\ell}$ is odd,
(b) $\phi$ is orthogonal if either (i)
$(\ell\bmod{4})=2$ and $x_{\ell-1}+x_{\ell}$ is even,
(ii) $(\ell\bmod{4})=0$, or (iii)
$\ell$ is odd and $x_{\ell-1}= x_{\ell}$,
and (c) $\phi$ is unitary if $\ell$ is odd and $x_{\ell-1}\neq x_{\ell}$.
(5) The representation $\phi$ is always orthogonal for
$\fg \in \{\fg_2,\ff_4,\fe_8\}$. (6) Assuming that $\fg\cong\fe_6$,  $\phi$ is orthogonal iff $x_1=x_6$ 
and $x_3=x_5$; it is unitary in all other cases. (7) For $\fg\cong\fe_7$, $\phi$ is symplectic 
iff $x_2+x_5+x_7$ is odd and orthogonal otherwise. 
\end{proposition}

\subsection{Alternating and Symmetric Square\label{def_alt_sym}}

The (inner) tensor product $\phi \otimes \psi$ of two representations $\phi$ and $\psi$ of a compact Lie algebra 
$\fg$ is defined as $(\phi \otimes \psi)(g):=
\phi(g) \otimes \unity_{\dim(\psi)} + \unity_{\dim(\phi)} \otimes \psi(g)$ for $g \in \fg$, or more explicitly as
$(\phi\otimes \psi)(g) [v_{i} \otimes w_{j}]:=
[\phi(g) v_{i}] \otimes w_{j} + v_{i} \otimes [\psi(g) w_{j}]
$ when acting on a basis
$\{ v_{i} \otimes w_{j}\,|\, i \in \{1,\ldots,\dim(\phi)\} \text{ and } j \in \{1,\ldots,\dim(\psi)\} \}$
of $\C^{\dim(\phi)}\otimes\C^{\dim(\psi)}$ (cf.\ Chap.~I, Sec.~3.2
of Ref.~\onlinecite{Bourb89}). Similarly, we define the tensor square $\phi^{\otimes 2}(g)
:=(\phi\otimes \phi)(g)=\phi(g) \otimes \unity_{\dim(\phi)} + \unity_{\dim(\phi)} \otimes \phi(g)$ and the tensor power
$\phi^{\otimes k}(g):=\phi(g) \otimes \unity_{\dim(\phi^{\otimes k-1})} + 
\unity_{\dim(\phi)} \otimes \phi^{\otimes k{-}1}(g)$. This  parallels the case of groups where the 
tensor product of representations is directly given by the Kronecker product of the representation matrices.

Let $\phi$ denote a $d$-dimensional representation of a compact Lie algebra $\fg$ acting on the vector space
$V:=\C^d$ with basis $\{v_{i} \,|\, i \in \{1,\ldots,d\}\}$. Recall that
$V^{\otimes 2}= V \otimes V$ decomposes into $\Alt^2 V \oplus \Sym^2 V$ where
$\Alt^2 V$ and $\Sym^2 V$ are spanned by 
$\{v_{i} \otimes v_{j} - v_{j} \otimes v_{i}\,|\, i, j \in \{1,\ldots,d\} \text{ and } i \neq j   \}$
and $\{v_{i} \otimes v_{j} + v_{j} \otimes v_{i}\,|\, i, j \in \{1,\ldots,d\}\}$, respectively.
The alternating square $\Alt^2 \phi:=\phi^{\otimes 2}|_{\Alt^2 V}$
of $\phi$ is defined by restricting $\phi^{\otimes 2}$ to $\Alt^2 V$. Accordingly,
the symmetric square $\Sym^2 \phi$ is given by $\phi^{\otimes 2}|_{\Sym^2 V}$.
The corresponding actions on basis vectors are 
\begin{align*}
(\Alt^2 \phi)(g)[v_{i}\otimes v_{j} - v_{j} \otimes v_{i}]
 & =  \left(
[\phi(g) v_{i}] \otimes v_{j} - v_{j} \otimes [\phi(g) v_{i}]
\right)\\
& \phantom{=\;} +  \left(
v_{i} \otimes [\phi(g) v_{j}] - [\phi(g) v_{j}] \otimes v_{i}
\right) \text{ and} \\
(\Sym^2 \phi)(g)[v_{i}\otimes v_{j} + v_{j} \otimes v_{i}]
 & =  \left(
[\phi(g) v_{i}] \otimes v_{j} + v_{j} \otimes [\phi(g) v_{i}]
\right)\\
& \phantom{=\;}
+  \left(
v_{i} \otimes [\phi(g) v_{j}] + [\phi(g) v_{j}] \otimes v_{i}
\right).
\end{align*}
Note that  $\dim(\Alt^2\phi)=\dim(\Alt^2V)=d(d{-}1)/2$ and $\dim(\Sym^2\phi)=\dim(\Sym^2V)=d(d{+}1)/2$.
In summary, we obtain the decomposition  $\phi^{\otimes 2}=\phi\otimes\phi=\Alt^2 \phi \oplus \Sym^2 \phi$.
First, we 
recall one elementary lemma.

\begin{lemma}\label{elem}
Consider two compact Lie algebras $\fh$ and $\fg$ with $\fh \subseteq \fg$ as well as
a faithful representation $\phi$ of $\fg$. 
We obtain $(\phi^{\otimes 2})|_{\fh}=(\phi |_{\fh})^{\otimes 2}$, 
$(\Alt^2 \phi)|_{\fh}\cong \Alt^2 (\phi |_{\fh})$,  and  $(\Sym^2 \phi)|_{\fh}\cong \Sym^2 (\phi |_{\fh})$.
\end{lemma}

\begin{proof}
We have $((\{\phi(g)| g \in \fg\})^{\otimes 2})|_{\fh}=(\{\phi(g)\otimes \unity +  \unity \otimes \phi(g)| g \in \fg\})|_{\fh}=
\{\phi(g)\otimes \unity +  \unity \otimes \phi(g)| g \in \fh\}=((\{\phi(g)| g \in \fg\})|_{\fh})^{\otimes 2}$. 
The rest follows along the same lines.
\end{proof}

Another elementary consequence of the definitions is the following lemma.
\begin{lemma}\label{plus}
Let $\phi$ and $\psi$ denote two  representations  of a compact Lie algebra $\fg$.\\
(i) $\Alt^2(\phi \oplus \psi)=\Alt^2 \phi \oplus (\phi \otimes \psi) \oplus \Alt^2 \psi$,\\
(ii) $\Sym^2(\phi \oplus \psi)=\Sym^2 \phi \oplus (\phi \otimes \psi) \oplus \Sym^2 \psi$,\\
(iii) $\phi$ is simple if $\Sym^2\phi$ is simple, and\\ 
(iv) $\phi$ is simple if $\Alt^2\phi$ is simple, $\phi$ is faithful, and $\fg$ is semisimple.
\end{lemma}
\begin{proof}
Statements (iii) and (iv) are a consequence of (i) and (ii) 
which both can be found on p.~473 of Ref.~\onlinecite{FH91}. 
\end{proof}

Let $\phi_1$ and   $\phi_2$ denote representations of the Lie algebras $\fg_1$
and $\fg_2$, respectively. The outer tensor product $\phi_1 \boxtimes \phi_2$ is defined as
$(\phi \boxtimes \psi)(g_1,g_2):=
\phi_1(g_1) \otimes \unity_{\dim(\phi_2)} + \unity_{\dim(\phi_1)} \otimes \phi_2(g_2)$ with $g_1 \in \fg_1$ and 
$g_2 \in \fg_2$. The statement of the following lemma can be deduced from
pp.~48--49  of Ref.~\onlinecite{Landsberg12} (see also p.~69 of Ref.~\onlinecite{Wolf68}).

\begin{lemma}\label{outertensor}
Assume that $\phi_i$ is a   representation of a compact Lie algebra $\fg_i$ where $i \in \{1,2\}$.\\
(i) $\Alt^2(\phi_1 \boxtimes \phi_2)
= (\Sym^2\phi_1 \boxtimes \Alt^2 \phi_2) \oplus (\Alt^2\phi_1 \boxtimes \Sym^2 \phi_2)$,\\
(ii) $\Sym^2(\phi_1 \boxtimes \phi_2)
= (\Sym^2\phi_1 \boxtimes \Sym^2 \phi_2) \oplus (\Alt^2\phi_1 \boxtimes \Alt^2 \phi_2)$.
\end{lemma}

A compact Lie algebra $\fg$ can be decomposed as 
$\fg\cong\oplus_k \fg_k$ where $\fg_k$ denotes a simple or abelian ideal of $\fg$. 
Moreover, 
all  simple representations of 
$\fg$ are of the form $\boxtimes_k \phi_k$ where $\phi_k$ denotes a simple representation
of $\fg_k$.

\begin{lemma}\label{irred_prod}
Consider a compact Lie algebra $\fg\cong\fg_1 \oplus \fg_2$. All finite-dimensional simple representations
$\phi$ of $\fg$ are of the form $\phi_1 \boxtimes \phi_2$ where $\phi_k$ denotes simple representations
of $\fg_k$. Furthermore, $\phi_1 \boxtimes \phi_2$ is a simple representation of $\fg$
for any  simple representations $\phi_k$ of $\fg_k$.
\end{lemma}

\begin{proof}
This follows directly from the corresponding fact for
(finite-dimensional) representations over $\C$ of unital, associative algebras over $\C$, see Sect.~12.1 
in Ref.~\onlinecite{Bourb12}.
\end{proof}

Recall that the trivial representation maps every 
element to the zero matrix of degree one, which is the only 
one-dimensional representation of semisimple Lie algebras.
We obtain as a  consequence of Lemmas~\ref{plus}--\ref{irred_prod} the following theorem.

\begin{theorem}\label{simple}
Let $\phi$ denote a faithful, semisimple representation of a compact  Lie algebra $\fg
\cong\fs \oplus \fc$ where  $\fs\cong\oplus_k \fs_k$ is semisimple (or zero)
and $\fc$ is abelian.
We assume that $\Alt^2\phi$ or $\Sym^2\phi$ is simple. It follows that $\phi |_{\fs}$ is simple, and 
$\phi |_{\fs_k}$ is distinct from the trivial representation
for at most one $k$. Moreover, the Lie algebra $\fs$ is simple (or zero).
\end{theorem}

\begin{proof}
As $\fg$ is compact, it is reductive and decomposes
as $\fg\cong\fs \oplus \fc$ where $\fs$ is semisimple and $\fc$ is abelian.
The representation $\phi |_{\fs}$ is simple by Lemma~\ref{plus}.
It follows from Lemma~\ref{irred_prod} that all simple representations of a direct sum of Lie algebras are
given by outer tensor products of simple representations.  Assuming that 
$\phi |_{\fs_k}$ is distinct from the trivial representation for more than one $k$ results in a contradiction
with Lemma~\ref{outertensor}. The remaining statement follows immediately, see also
pp.~27--28 and p.~321 of Ref.~\onlinecite{GG78} for information on irreducible subalgebras.
\end{proof}

As the decomposition $\fg\cong\fs \oplus \fc$ of a compact (or reductive) Lie algebra will appear
often, we recall that  $[\fg,\fg]:=\{ [g_1,g_2] \,|\, g_1,g_2 \in \fg\}$ denotes an 
ideal (and a subalgebra) of the Lie algebra $\fg$.
If $\fg$ is reductive (or compact), $[\fg,\fg]$ is a semisimple subalgebra of $\fg$ and we denote it by $\fs(\fg):=[\fg,\fg]$.
Moreover, the center  of $\fg$ is given by $\fc(g):=\{g_1 \in \fg \,|\, [g_1,g_2]=0\, \text{ for all }\, g_2 \in \fg\}$.
We obtain that $\fg = \fs(\fg) \oplus  \fc(\fg)$ if $\fg$ is reductive (or compact), where 
$\fs = \fs(\fg)$ and $\fc = \fc(\fg)$.
Theorem~\ref{simple} can be used
to simplify the discussion by limiting the search for simple, alternating and symmetric squares to
simple representations of compact simple Lie algebras. 

\section{Some techniques of Dynkin\label{Dynkin_techniques}}

We recall some elementary but powerful techniques of Dynkin
(see Sec.~3 of Ref.~\onlinecite{Dynkin57}) which allow us to analyze the decomposition
of a tensor product of representations very efficiently. In particular, 
we discuss \emph{minimal chains}, the notion of \emph{subordination}, and 
the \emph{method of parts}. 
At the end of this appendix, some elementary facts about the \emph{adjoint representation} are 
also recalled.
Combinations of these techniques will be directly applied, e.g., in Appendix~\ref{case-by-case}.

{
\fontsize{10}{12}
\renewcommand{\baselinestretch}{1}
\selectfont
\begin{table}[tb]
\caption{\label{car_tab}Cartan matrices $\alpha_{\fg}$ corresponding to the root system of $\fg_{\C}$}
\begin{ruledtabular}
\begin{tabular}{c}
$
\alpha_{\su(\ell{+}1)}=
\left(
\begin{smallmatrix}
2 & -1 & 0 & \cdot\cdot & 0 & 0\\
-1 & 2 & -1 & \ddotdot & \vdotdot & \vdotdot\\
0 & -1 & \ddotdot  & \ddotdot & 0 & 0\\
\vdotdot & \ddotdot & \ddotdot & 2 & -1 & 0\\
0 & \cdot\cdot & 0 & -1 & 2 & -1\\
0 & \cdot\cdot & 0 & 0 & -1 & 2\\
\end{smallmatrix}
\right)
$,
$\alpha_{\so(2\ell{+}1)}=
\left(
\begin{smallmatrix}
2 & -1 & 0 & \cdot\cdot & 0 & 0\\
-1 & 2 & -1 & \ddotdot & \vdotdot & \vdotdot\\
0 & -1 & \ddotdot  & \ddotdot & 0 & 0\\
\vdotdot & \ddotdot & \ddotdot & 2 & -1 & 0\\
0 & \cdot\cdot & 0 & -1 & 2 & -2\\
0 & \cdot\cdot & 0 & 0 & -1 & 2\\
\end{smallmatrix}
\right)
$,
$\alpha_{\usp(\ell)}=
\left(
\begin{smallmatrix}
2 & -1 & 0 & \cdot\cdot & 0 & 0\\
-1 & 2 & -1 & \ddotdot & \vdotdot & \vdotdot\\
0 & -1 & \ddotdot  & \ddotdot & 0 & 0\\
\vdotdot & \ddotdot & \ddotdot & 2 & -1 & 0\\
0 & \cdot\cdot & 0 & -1 & 2 & -1\\
0 & \cdot\cdot & 0 & 0 & -2 & 2\\
\end{smallmatrix}
\right)
$,\\[5mm]
$\alpha_{\so(2\ell)}=
\left(
\begin{smallmatrix}
2 & -1 & 0 & \cdot\cdot & 0 & 0\\
-1 & 2 & -1 & \ddotdot & \vdotdot & \vdotdot\\
0 & -1 & \ddotdot  & \ddotdot & 0 & 0\\
\vdotdot & \ddotdot & \ddotdot & 2 & -1 & -1\\
0 & \cdot\cdot & 0 & -1 & 2 & 0\\
0 & \cdot\cdot & 0 & -1 & 0 & 2\\
\end{smallmatrix}
\right)
$,
$\alpha_{\fg_2}=
\left(
\begin{smallmatrix}
2 & -1\\
-3 & 2
\end{smallmatrix}
\right)
$,
$\alpha_{\ff_4}=
\left(
\begin{smallmatrix}
2 & -1 & 0  & 0\\
-1 & 2 & -2 & 0\\
0 & -1 & 2 & -1\\
0 & 0 & -1 & 2
\end{smallmatrix}
\right)
$,
$\alpha_{\fe_6}=
\left(
\begin{smallmatrix}
2 & 0 & -1 & 0 & 0 & 0\\
0 & 2 & 0 & -1 &  0 & 0\\
-1 & 0 & 2  & -1 & 0 & 0\\
0 & -1 & -1 & 2 & -1 & 0\\
0 & 0 & 0 & -1 & 2 & -1\\
0 & 0 & 0 & 0 & -1 & 2\\
\end{smallmatrix}
\right)
$,\\[5mm]
$\alpha_{\fe_7}=
\left(
\begin{smallmatrix}
2 & 0 & -1 & 0 & 0 & 0 & 0\\
0 & 2 & 0 & -1 &  0 & 0 & 0\\
-1 & 0 & 2  & -1 & 0 & 0 & 0\\
0 & -1 & -1 & 2 & -1 & 0 & 0\\
0 & 0 & 0 & -1 & 2 & -1 & 0\\
0 & 0 & 0 & 0 & -1 & 2 & -1\\
0 & 0 & 0 & 0 & 0 & -1 & 2\\
\end{smallmatrix}
\right)
$,
$\alpha_{\fe_8}=
\left(
\begin{smallmatrix}
2 & 0 & -1 & 0 & 0 & 0 & 0 & 0\\
0 & 2 & 0 & -1 &  0 & 0 & 0 & 0\\
-1 & 0 & 2  & -1 & 0 & 0 & 0 & 0\\
0 & -1 & -1 & 2 & -1 & 0 & 0 & 0\\
0 & 0 & 0 & -1 & 2 & -1 & 0 & 0\\
0 & 0 & 0 & 0 & -1 & 2 & -1 & 0\\
0 & 0 & 0 & 0 & 0 & -1 & 2 & -1\\
0 & 0 & 0 & 0 & 0 & 0 & -1 & -2\\
\end{smallmatrix}
\right)
$
\end{tabular}
\end{ruledtabular}
\end{table}
}

As discussed in Appendix~\ref{compact_Lie}, the highest weight $x_{\phi}=(x_{\phi,1},\ldots,x_{\phi,\ell})$ characterizes 
an equivalence class corresponding to a simple representation $\phi$ 
of a compact semisimple Lie algebra $\fg$ of rank $\ell$. 
The root system $R$ of $\fg_{\C}$ (see Chap.~VIII, Sec.~2 of Ref.~\onlinecite{Bourb08b} and Chap.~VI of 
Ref.~\onlinecite{Bourb08a}) has a basis $\{\alpha_1,\ldots,\alpha_\ell\}$ of so-called 
\emph{simple roots} 
$\alpha_a=(\alpha_{a,1},\ldots,\alpha_{a,\ell})$ such that the corresponding Cartan matrix $\alpha_{\fg}$ 
has entries $\alpha_{a,b}$ (see Table~\ref{car_tab}). The Dynkin diagrams in Table~\ref{dyn_tab}
can be recovered from the Cartan matrices: there is a link between nodes $a$ and $b$ if $\alpha_{a,b}\neq 0$,
and its type is then given by the value of $\alpha_{a,b}/\alpha_{b,a} \in \{1,2,3,1/2,1/3\}$ (see Chap.~VI, Sec.~4.2 of 
Ref.~\onlinecite{Bourb08a}). The basis of simple roots
introduces a partial order (e.g.,) on  the highest weights: $x_{\phi} \leq x_{\psi}$ if $x_{\psi}-x_{\phi}=
\sum_{j=1}^{\ell} n_j \alpha_j$ for some non-negative integers $n_j$. 
Using this notation, one obtains the following lemma 
(see  Chap.~VIII, Sec.~7.4, Prop.~9(ii) of Ref.~\onlinecite{Bourb08b}).

\begin{lemma}\label{tensor_rule}
Consider a compact semisimple Lie algebra $\fg$ and two of its  simple 
representations $\phi$ and $\psi$. The tensor product 
$\phi\otimes\psi\cong \oplus_k \tau_k$ decomposes into simple representations $\tau_k$
with $x_{\tau_k}\leq x_{\phi}+x_{\psi}$
\end{lemma}

Let $(\cdot | \cdot)$ denote a non-degenerate symmetric bilinear form on the real vector space generated by
the root system $R$ where $\alpha_{a,b}=2 (\alpha_a | \alpha_b)/  (\alpha_b | \alpha_b)$ and  
$x_{a}=2 (x | \alpha_a)/(\alpha_a | \alpha_a)$ (see pp.~157--158 and 180--181 of Ref.~\onlinecite{Bourb08a}). 
The explicit form of $(\cdot | \cdot)$
will not be important in the following, but we emphasize that $(\alpha_a | \alpha_b) =0 \Leftrightarrow
\alpha_{a,b}=0 \Leftrightarrow \alpha_{b,a}=0$.
Now, more detailed information on the tensor product decomposition can be obtained by 
applying the notion of  a \emph{minimal chain} linking two highest weights (see Sec.~3.1 of Ref.~\onlinecite{Dynkin57}, 
Chap.~VIII, Sec.~7, Ex.~18 of Ref.~\onlinecite{Bourb08b}, or Ref.~\onlinecite{Aslaksen94}).

\begin{definition}
Consider two highest weights $x_{\phi}$ and $x_{\psi}$ corresponding to two simple representations
$\phi$ and $\psi$ of a compact semisimple Lie algebra $\fg$. A sequence
$[\alpha_{j_1},\ldots,\alpha_{j_n}]$ of simple roots of $\fg_{\C}$ is said to 
be a \emph{chain} joining
$x_{\phi}$ and $x_{\psi}$ if 
(a) $( x_{\phi} | \alpha_{j_1}) \neq 0$,
(b) $( \alpha_{j_k} | \alpha_{j_{k+1}}) \neq 0$  for $k \in \{1,\ldots,n{-}1\}$,
and (c) $( \alpha_{j_n}  |  x_{\psi}) \neq 0$. Such a chain 
is called \emph{minimal} if no proper subsequence joins $x_{\phi}$ and $x_{\psi}$, i.e.,
(d) $( x_{\phi} | \alpha_{j_k}) = 0$ for $k>1$,
(e) $( \alpha_{j_k} | \alpha_{j_{h}}) = 0$  for $h>k+1$, and
(f) $( \alpha_{j_k}  |  x_{\psi}) = 0$ for $k<n$.
\end{definition}

Using a minimal chain, we can restrict the possible highest weights in the decomposition
of a tensor product.

\begin{proposition}
If $x_{\tau}\neq x_{\phi}+x_{\psi}$ is the highest weight of a simple representation $\tau$
occurring in the decomposition of the tensor product $\phi\otimes\psi$
of two simple representations $\phi$ and $\psi$, then 
$x_{\tau}\leq x_{\phi}+x_{\psi}-\sum_{k=1}^{n} \alpha_{j_k}$
for some minimal chain $[\alpha_{j_1},\ldots,\alpha_{j_n}]$ linking $x_{\phi}$ and $x_{\psi}$.
In particular,  both $x_{\phi}+x_{\psi}-\sum_{k=1}^{n} \alpha_{j_k}$ and
$x_{\phi}+x_{\psi}$ appear in the decomposition. 
\end{proposition}

We follow Dynkin (Sec.~3.3 in Ref.~\onlinecite{Dynkin57})
and introduce the powerful notion of \emph{subordination}, see also Ref.~\onlinecite{Aslaksen94} and
Chap.~VIII, Sec.~7, Ex.~16 of Ref.~\onlinecite{Bourb08b}.

\begin{definition}
We consider the  representations $\phi\cong\oplus_{j=1}^r \phi_j$ and 
$\psi\cong\oplus_{j=1}^s \psi_j$
of a compact simple Lie algebra which decompose into simple representations $\phi_j$ and $\psi_j$.
Let $x_\tau=(x_{\tau,1},\ldots,x_{\tau,\ell})$  denote the highest weight
of a simple representation $\tau$. Assuming that both $\phi$ and $\psi$ are simple
(i.e., $r=s=1$), $\phi$ is called \emph{subordinate} to $\psi$ 
(written $\phi \sqsubseteq \psi$ or $x_\phi \sqsubseteq x_\psi$) if $x_{\phi,k} \leq x_{\psi,k}$ for 
all $k\in \{1,\ldots,\ell\}$. Otherwise, we call $\phi$ subordinate to $\psi$ if $r\leq s$ and
there exists a re-ordering of the simple components such that $\phi_j \sqsubseteq \psi_j$ for $j \in \{1,\ldots,
r\}$.
\end{definition}

Important consequences of  subordination are summarized in the following proposition where $m(\phi,\psi)$ (or $m(x_\phi,\psi)$)
denotes the multiplicity of the simple representation $\phi$ in the decomposition of the representation $\psi$.

\begin{proposition}\label{prop_subord}
Let $\phi$, $\phi_j$, $\psi$, and $\psi_j$ denote simple representations of 
a compact simple Lie algebra $\fg$ such that $\phi \sqsubseteq \psi$ and 
$\phi_j \sqsubseteq \psi_j$ for $j\in\{1,2\}$.
Given a basis $\{\alpha_1,\ldots,\alpha_\ell\}$ of simple roots for $\fg_{\C}$ and
a set $\{n_1,\ldots,n_\ell\}$ of non-negative integers, we obtain:\\
(i) $\phi_1 \otimes \phi_2 \sqsubseteq \psi_1 \otimes \psi_2$.\\
(ii) $m(x_{\phi_1}+x_{\phi_2}-\sum_{k=1}^\ell n_k\, \alpha_{k},\phi_1\otimes\phi_2)  
\leq m(x_{\psi_1}+x_{\psi_2}-\sum_{k=1}^\ell n_k\, \alpha_{k},\psi_1\otimes\psi_2)$.\\
(iii) $m(2 x_{\phi}-\sum_{k=1}^\ell n_k\, \alpha_k,\Sym^2\phi)  
\leq m(2 x_{\psi}-\sum_{k=1}^\ell n_k\, \alpha_k,\Sym^2\psi)$.\\
(iv) $m(2 x_{\phi}-\sum_{k=1}^\ell n_k\, \alpha_k,\Alt^2\phi)  
\leq m(2 x_{\psi}-\sum_{k=1}^\ell n_k\, \alpha_k,\Alt^2\psi)$.
\end{proposition}

\begin{proof}
Part (i) is Thm.~3.17 of Ref.~\onlinecite{Dynkin57}. A proof of the other parts is given in Thm.~1 of 
Ref.~\onlinecite{Aslaksen94}.
\end{proof}

Consider the basis $\{\alpha_1,\ldots,\alpha_\ell\}$ of simple roots 
of $\fg_{\C}$ for a compact semisimple Lie algebra $\fg$.
A proper subset of the simple roots obtained by
deleting the simple roots $\alpha_j$ with  $j \in D\subsetneq \{1,\ldots,\ell\}$
generates a smaller root system 
corresponding to a semisimple subalgebra $\fh$ of $\fg$ (resp.\ $\fh_{\C}$ of $\fg_{\C}$), 
cf.\ Chap.~VIII, Sec.~3.1 of Ref.~\onlinecite{Bourb08b}. This is equivalent to deleting the nodes
in $D$ from the Dynkin diagram or shortening the highest weight $x$ of a simple representation
by deleting the entries $x_j$ with $j \in D$.
This allows us to introduce the notion of parts of highest weights 
(see Sec.~3.2 of Ref.~\onlinecite{Dynkin57} and Chap.~XIV of Ref.~\onlinecite{Cahn84}).

\begin{definition}
Consider the  representations $\phi\cong\oplus_{j=1}^r \phi_j$ and 
$\psi\cong\oplus_{j=1}^s \psi_j$ of the compact semsimple Lie algebras $\fh$ and $\fg$,
respectively. In addition, $\phi_j$ and $\psi_j$ denote simple representations and $\ff$ is 
a subalgebra of $\fg$. Assuming that both $\phi$ and $\psi$ are simple (i.e., $r=s=1$),
$\phi$ is called a \emph{part} of $\psi$ (written $\phi \subseteq \psi$ or $x_\phi \subseteq x_\psi$) 
if $\phi$ is faithful and $x_\phi$ is a shortened version of 
$x_\psi$ which is obtained by deleting some (but not all) or none of the entries $x_{\psi,j}$.
Otherwise, $\phi$ is a part of $\psi$ if $r\leq s$ and
there exists a re-ordering of the simple components such that $\phi_j \subseteq \psi_j$ for $j \in \{1,\ldots,r\}$.
\end{definition}

We obtain the following proposition which is essential in applying the method of parts.

\begin{proposition}\label{prop_part}
Consider the semisimple subalgebra $\fh$ of the compact semisimple Lie algebra $\fg$.
Let $\phi$ and $\phi_j$ (resp.\ $\psi$ and $\psi_j$) denote  simple representations of $\fh$ 
(resp.\ $\fg$). Assume that $\phi \subseteq \psi$ and $\phi_j \subseteq \psi_j$. 
Given a basis $\{\alpha_1,\ldots,\alpha_\ell\}$ of the simple roots of $\fg_{\C}$ and
a set $\{n_1,\ldots,n_\ell\}$ of non-negative integers, we obtain:\\
(i) $\phi_1 \otimes \phi_2 \subseteq \psi_1 \otimes \psi_2$.\\
(ii) $m(x_{\phi_1}+x_{\phi_2}-\sum_{k=1}^\ell n_k\, \alpha_{k},\phi_1\otimes\phi_2)  
\leq m(x_{\psi_1}+x_{\psi_2}-\sum_{k=1}^\ell n_k\, \alpha_{k},\psi_1\otimes\psi_2)$.\\
(iii) $m(2 x_{\phi}-\sum_{k=1}^\ell n_k\, \alpha_k,\Sym^2\phi)  
\leq m(2 x_{\psi}-\sum_{k=1}^\ell n_k\, \alpha_k,\Sym^2\psi)$.\\
(iv) $m(2 x_{\phi}-\sum_{k=1}^\ell n_k\, \alpha_k,\Alt^2\phi)  
\leq m(2 x_{\psi}-\sum_{k=1}^\ell n_k\, \alpha_k,\Alt^2\psi)$. 
\end{proposition}

\begin{proof}
Part (i) is Thm.~3.10 of Ref.~\onlinecite{Dynkin57}. Part (ii)
follows from (i) and Lemma~\ref{tensor_rule} by applying Thms.~3.7 and 3.9 of Ref.~\onlinecite{Dynkin57}.
We also recommend the discussion in
Chap.~XIV of Ref.~\onlinecite{Cahn84}.
A proof of (ii)-(iv) is given in Thm.~1 of Ref.~\onlinecite{Kraemer78}.
\end{proof}

Finally, recall that the adjoint representation of a Lie algebra $\fg$  maps each element $g_1\in\fg$
to the endomorphism $\ad_{\fg}(g_1)$ which is defined as $(\ad_{\fg}(g_1))(g_2):=[g_1,g_2]$
for each $g_2 \in \fg$.
We summarize some facts about the adjoint representation of a simple Lie algebra.
These results are essential for dealing with self-dual representations.

\begin{lemma}\label{adjoint}
Consider a compact real (or complex) Lie algebra $\fg$.
(i) The adjoint representation of $\fg$ restricted to a proper subalgebra
is reducible.
(ii) The adjoint representation of $\fg$ is simple if $\fg$ is simple.
(iii) The adjoint representation of $\so(k)$ with $k\geq 3$ is isomorphic to the alternating 
square of the standard representation; moreover $\Alt^2 \phi_{(2)}=\phi_{(2)}$ for $\so(3)\cong\su(2)$,
$\Alt^2\phi_{(1,1)}=\phi_{(2,2)}$ for $\so(4)$,
$\Alt^2\phi_{(1,0)}=\phi_{(0,2)}$ for $\so(5)$, $\Alt^2\phi_{(1,0,0)}=\phi_{(0,1,1)}$ for $\so(6)$, and
$\Alt^2\phi_{(1,0,\ldots,0)}=\phi_{(0,1,0,\ldots,0)}$ for $k\geq 7$.
(iv)~The adjoint representation of $\usp(\ell)$ with $\ell \geq 1$ is isomorphic to the symmetric square
$\Sym^2 \phi_{(1,0,\ldots,0)}=\phi_{(2,0,\ldots,0)}$ of the standard representation
$\phi_{(1,0,\ldots,0)}$. \end{lemma} 

\begin{proof}
As the Lie commutator $[\cdot,\cdot]$ is closed for subalgebras,
(i) follows. Assuming that $\fg$ is simple we obtain (ii) as $\fg$ has no proper ideals and the 
adjoint representation has to be simple.
The fact that the alternating square of the standard representation
is isomorphic to the adjoint representation for $\so(k)$ with $k>1$ can be found on
p.~353 of Ref.~\onlinecite{Dynkin57}  or pages~199 and 213 of Ref.~\onlinecite{Bourb08b}.
The statement for $\usp(\ell)$ is given on p.~353  of Ref.~\onlinecite{Dynkin57} and in 
Chap.~VIII, Sec.~13, Ex.~8 of Ref.~\onlinecite{Bourb08b}.
Most of the corresponding highest weights are (e.g.,) given in Table~28 of 
Ref.~\onlinecite{Dynkin57}, and the remaining ones statements
can be directly verified.
\end{proof}

\section{Case-by-case analysis for Alternating and Symmetric Squares\label{case-by-case}}

In this appendix, we build on Appendix~\ref{Dynkin_techniques} and provide a streamlined version of
the laborious details needed for the proofs of 
Theorems~\ref{thm_alt_square} and \ref{thm_sym_square}.

\subsection{Alternating Squares\label{AltSq}}

In order to limit the case-by-case discussion, we reproduce a theorem and  proof of 
Dynkin (see Thm.~4.5 of Ref.~\onlinecite{Dynkin57}) 
which completely characterizes the irreducibility of the alternating squares of self-dual representations.

\begin{theorem}[Dynkin]\label{alt_self}
Let $\phi$ denote a self-dual and faithful representation of a compact semisimple Lie algebra $\fg$.
The representation $\Alt^2\phi$ is not simple, if the pair $(\fg,x_\phi)$ is not contained in the
following list: (a)~$(\su(2),(1))$, (b)~$(\su(2),(2))$, (c)~$(\so(2\ell+1),(1,0,\ldots,0))$ with $\ell \ge 2$, and
(d)~$(\so(2\ell),(1,0,\ldots,0))$ with $\ell \ge 3$.
\end{theorem}

\begin{proof}
Due to Theorem~\ref{simple} of Appendix~\ref{app_pre}, we can restrict us to simple representations $\phi$ of 
compact simple Lie algebras
$\fg$. As $\phi$ is self-dual, we get either (i) $\phi(\fg)\subseteq \usp(\dim(\phi)/2)$ or 
(ii) $\phi(\fg)\subseteq \so(\dim(\phi))$. We consider first the case (i): Given the 
standard representation $\psi$ of $\usp(\ell)$ with $x_\psi=(1,0,\ldots,0)$
and $\ell > 1$, we obtain that $\Alt^2\psi=\tau_2 \oplus \tau_0$ decomposes
into representations with highest weights $x_{\tau_2}=(0,1,0,\ldots,0)$ and $x_{\tau_0}=(0,\ldots,0)$
(see p.~360 in Ref.~\onlinecite{Dynkin57}, pp.~260--261 of Ref.~\onlinecite{FH91}, or pp.~206--207 of 
Ref.~\onlinecite{Bourb08b}).
It follows from Lemma~\ref{elem} of Appendix~\ref{app_pre} that $\Alt^2\phi$ cannot be simple for $\dim(\phi)\neq 2$. 
For $\dim(\phi)= 2$, we get the element (a) from the list. Let us now treat the case (ii):
In order to apply Lemma~\ref{elem} of Appendix~\ref{app_pre} again, we analyze the standard 
representation $\psi$ where 
$x_{\psi}=(1,0,\ldots,0)$
for the Lie algebras $\so(k)$ with $k\geq 5$, $x_{\psi}=(2)$ for $\so(3)\cong\su(2)$, and $x_\psi=
(1,1)$ for $\so(4)\cong\su(2)\oplus\su(2)$. 
Using Lemma~\ref{adjoint} of Appendix~\ref{Dynkin_techniques}, one obtains that $\Alt^2\psi$ is isomorphic 
to the adjoint representation of the corresponding Lie algebras and that the adjoint representation
is no longer simple when restricted to a proper subalgebra. Therefore, either $\fg=\so(2\ell+1)$ or $\fg=\so(2\ell)$
and we obtain the remaining elements in the list, while  $\so(4)\cong\su(2)\oplus\su(2)$ is 
not simple and does not lead to a simple alternating square.
\end{proof}

Consequently, we can limit us now 
to representations which are not self-dual. Due to Proposition~\ref{Malcev} of Appendix~\ref{app_pre}, we get
that either $\fg=\su(\ell+1)$ with $\ell >1$, $\fg=\so(4k+2)$, or $\fg=\fe_6$. The following lemma simplifies
our search significantly.

\begin{lemma}\label{simply_laced_alt}
Let $\fg$ denote a compact simple Lie algebra which is  simply-laced, i.e., $\su(\ell+1)$, $\so(2\ell)$, 
 $\fe_6$, $\fe_7$, or $\fe_8$. We consider the  simple representations $\phi$, $\tau$, and $\psi$
where $\tau$ and $\psi$ are representations of $\fg$ and $\phi$ is a representation of a subalgebra
$\fh\subseteq\fg$. We assume that $\phi \subseteq \tau \sqsubseteq \psi$, i.e., that
$\tau$ is subordinate to $\psi$ and that $\phi$ is a part of $\tau$.
The representation $\Alt^2\psi$ is not simple if one of the following
conditions holds:\\
(i) $x_\phi=(1,1)$ and either $\fh=\su(3)$ or $\fh=\su(2)\oplus\su(2)$.\\
(ii) More than one component $x_{\psi,j}$ of the highest weight $x_\psi$ is non-zero.\\
(iii) $x_\phi=(3)$ and $\fh=\su(2)$.\\
(iv) The highest weight $x_\psi$ has a component which is larger than two.\\
(v) $x_\phi=(0,2,0)$ and $\fh=\su(4)$.\\
(vi) The highest weight $x_\psi$ has a component which is larger than one and lies next to any 
end of the Dynkin diagram.\\
(vii) $x_\phi=(0,0,1,0,0)$ and $\fh=\su(6)$.\\
(viii) The highest weight $x_\psi$ has a non-zero component which is not at or next to any end of the Dynkin diagram.
\end{lemma}

\begin{proof}
The non-irreducibility of $\Alt^2\phi$ for the cases in
the statements (i), (iii), (v), and (vii) can be verified by direct computations.
Then, all other statements follow via Propositions~\ref{prop_subord}
and \ref{prop_part} of Appendix~\ref{Dynkin_techniques}.
\end{proof}

Using Lemma~\ref{simply_laced_alt} we can completely treat the case of simple representations
which are not self-dual.

{
\fontsize{10}{12}
\renewcommand{\baselinestretch}{1}
\selectfont
\begin{table}[tb]
\caption{\label{alt_square_nonselfdual}Irreducible representations which are not self-dual and 
whose respective alternating square is irreducible (Dynkin)}
\begin{ruledtabular}
\begin{tabular}[t]{l@{\hspace{4pt}}l@{\hspace{4pt}}l@{\hspace{4pt}}l@{\hspace{1pt}}l@{\hspace{4pt}}l@{\hspace{1pt}}l}
case & $\fg$ & $\ell$ & $\phi$  & dim$(\phi)$ & $\Alt^2 \phi$ & dim$(\Alt^2 \phi)$\\[1mm] \hline\\[-3.5mm]
(i) & $\su(\ell+1)$ & $\ell\geq 2$ &
$(1,0,\ldots,0)$ & $\ell{+}1$ & $(0,1,0,\ldots,0)$ & $\tfrac{\ell(\ell+1)}{2}$\\ 
(ii) & $\su(\ell+1)$ & $\ell\geq 2$ &
$(2,0,\ldots,0)$ & $\tfrac{(\ell+1)(\ell+2)}{2}$ & $(2,1,0,\ldots,0)$ & $3 \tbinom{\ell+3}{4}$\\ 
(iii) & $\su(\ell+1)$ & $\ell\geq 3$ &
$(0,1,0,\ldots,0)$ & $\tfrac{\ell(\ell+1)}{2}$ & $(1,0,1,0,\ldots,0)$ & $3 \tbinom{\ell+2}{4}$\\ 
(iv) & $\so(10)$ & -- &
$(0,0,0,1,0)$ & $16$ & $(0,0,1,0,0)$ & $120$\\ 
(v) & $\fe_6$ & -- &
$(1,0,0,0,0,0)$ & $27$ & $(0,0,1,0,0,0)$ & $351$\\ 
\end{tabular}
\end{ruledtabular}
\end{table}
}

\begin{lemma}\label{lemma_selfdual_alt}
Let $\phi$ denote a faithful representation of a compact semisimple Lie algebra
$\fg$ such that $\phi$ is not self-dual and that the alternating square $\Alt^2\phi$ is simple.
All possible cases (up to outer automorphisms of $\fg$)  are given in Table~\ref{alt_square_nonselfdual}.
\end{lemma}

\begin{proof}
Due to  Theorem~\ref{simple} of Appendix~\ref{app_pre}, $\phi$ is simple and $\fg$ is simple.
Proposition~\ref{Malcev} of Appendix~\ref{app_pre} limits the possible cases to (a) $\fg=\su(\ell+1)$ with $\ell >1$, 
(b) $\fg=\so(4k+2)$ with $k \ge 1$, or (c) $\fg=\fe_6$.
Due to Lemma~\ref{simply_laced_alt} only the cases (i)-(iii) of Table~\ref{alt_square_nonselfdual} can appear
for (a). The irreducibility of the alternating square in the case (i) is well known
(see p.~225 of Ref.~\onlinecite{FH91} or p.~192 of Ref.~\onlinecite{Bourb08b}). The case (ii)-(iii) are discussed (e.g.,) in
Ex.~15.33 and 15.32 on pp.~226--227 of Ref.~\onlinecite{FH91}. This completes (a).
For (b), one can directly verify that $\Alt^2 \phi_{(0,0,0,0,0,1,0)}$ is not simple if $\fg=\so(14)$. 
We apply Propositions~\ref{prop_subord} and \ref{prop_part}  of Appendix~\ref{Dynkin_techniques} 
and can now limit us to $\so(10)$, 
while $\so(6)\cong\su(4)$ 
has been already discussed. It can be directly checked that only the case (iv) of 
Table~\ref{alt_square_nonselfdual} remains. The Lie algebra $\fg=\fe_6$ from case (c) can also treated by 
explicit computations leading only to case (v) of 
Table~\ref{alt_square_nonselfdual}, which completes the proof.
\end{proof}

The case-by-case analysis of Theorem~\ref{alt_self} and Lemma~\ref{lemma_selfdual_alt}
imply Theorem~\ref{thm_alt_square}.

\subsection{Symmetric Squares\label{SymSq}}

In order to limit the case-by-case discussion, we reproduce a theorem of 
Dynkin (see Thm.~4.5 of Ref.~\onlinecite{Dynkin57}) which completely treats the case of self-dual
representations.

\begin{theorem}[Dynkin]\label{sym_self}
Let $\phi$ denote a self-dual and faithful representation of a compact semisimple Lie algebra $\fg$.
The representation $\Sym^2\phi$ is not simple, if the pair $(\fg,x_\phi)$ is not equal to $(\usp(\ell),(1,0,\ldots,0))$.
\end{theorem}

\begin{proof}
Due to Theorem~\ref{simple}, we can limit us to simple representations $\phi$ of compact simple Lie algebras. 
As $\phi$ is self-dual, we get either  (i) $\phi(\fg)\subseteq \so(\dim(\phi))$ or
(ii) $\phi(\fg)\subseteq \usp(\dim(\phi)/2)$. We consider first the case (i): The Lie algebras $\so(\dim(\phi))$ with 
$\dim(\phi)\in\{1,2\}$ do not have simple subalgebras. The Lie algebra $\so(3)\cong\su(2)\cong\usp(1)$
has no proper simple subalgebra, but leads to the case with $\fg=\usp(1)$ and $x_\phi=(1)$. The
symmetric square $\Sym^2\psi_{(1,1)}=\psi_{(2,2)}\oplus\psi_{(0,0)}$ for the standard representation
$\psi_{(1,1)}$ of $\so(4)\cong\su(2)\oplus\su(2)$ splits up into two components. So, $(\Sym^2\psi_{(1,1)})|_{\fg}$
is also not simple. The symmetric square $\Sym^2 \psi_{(1,0,\ldots,0)}=
\psi_{(2,0,\ldots,0)}\oplus\psi_{(0,\ldots,0)}$ for the standard representation
$\psi_{(1,0,\ldots,0)}$ of $\so(\dim(\phi))$ with  $\dim(\phi)\geq 5$ decomposes 
(see Ex.~19.21 of Ref.~\onlinecite{FH91}), and no further cases can appear for (i). We analyze the case (ii):
The standard representation $\psi:=\psi_{(1,0,\ldots,0)}$ of $\usp(\ell)$ has the symmetric
square $\Sym^2\psi=\psi_{(2,0,\ldots,0)}$ which is simple and isomorphic to its adjoint representation
and $\Sym^2\psi$ decomposes when
restricted to a proper subalgebra  (see Lemma~\ref{adjoint} of Appendix~\ref{Dynkin_techniques}). 
It follows that
$\fg=\usp(\ell)$ which completes the proof.
\end{proof}

Consequently, we can limit us again
to representations which are not self-dual.
Due to Proposition~\ref{Malcev} of Appendix~\ref{app_pre}, we have
either $\fg=\su(\ell+1)$ with $\ell >1$, $\fg=\so(4k+2)$, or $\fg=\fe_6$. 
As in Section~\ref{AltSq}, we can simplify our analysis.

\begin{lemma}\label{simply_laced_sym}
Let $\fg$ denote a compact simple Lie algebra which is  simply-laced, i.e., $\su(\ell+1)$, $\so(2\ell)$, 
 $\fe_6$, $\fe_7$, or $\fe_8$. We consider the simple representations $\phi$, $\tau$, and $\psi$
where $\tau$ and $\psi$ are representations of $\fg$ and $\phi$ is a representation of a subalgebra
$\fh\subseteq\fg$. We assume that $\phi \subseteq \tau \sqsubseteq \psi$, i.e., that
$\tau$ is subordinate to $\psi$ and that $\phi$ is a part of $\tau$.
The representation $\Sym^2\psi$ is not simple if one of the following
conditions holds:\\
(i) $x_\phi=(1,1)$ and either $\fh=\su(3)$ or $\fh=\su(2)\oplus\su(2)$.\\
(ii) More than one component $x_{\psi,j}$ of the highest weight $x_\psi$ is non-zero.\\
(iii) $x_\phi=(2)$ and $\fh=\su(2)$.\\
(iv) The highest weight $x_\psi$ has a component which is larger than one.\\
(v) $x_\phi=(0,1,0)$ and $\fh=\su(4)$.\\
(vi) The highest weight $x_\psi$ has a non-zero component which is not at any end of the Dynkin diagram.
\end{lemma}

\begin{proof}
The non-irreducibility of $\Sym^2\phi$ for the cases in
the statements (i), (iii), and (v) can be verified by direct computations.
Then, all other statements follow via Propositions~\ref{prop_subord}
and \ref{prop_part} of Appendix~\ref{Dynkin_techniques}.
\end{proof}

Using Lemma~\ref{simply_laced_sym} we can completely treat the case of simple representations
which are not self-dual.

{
\fontsize{10}{12}
\renewcommand{\baselinestretch}{1}
\selectfont
\begin{table}[tb]
\caption{\label{sym_square_nonselfdual}Irreducible representations which are not self-dual and whose respective symmetric
square is irreducible}
\begin{ruledtabular}
\begin{tabular}[t]{l@{\hspace{4pt}}l@{\hspace{4pt}}l@{\hspace{4pt}}l@{\hspace{1pt}}l@{\hspace{4pt}}l@{\hspace{1pt}}l}
case & $\fg$ & $\ell$ & $\phi$  & dim$(\phi)$ & $\Sym^2 \phi$ & dim$(\Sym^2 \phi)$\\[1mm] \hline\\[-3.5mm]
(i) & $\su(\ell+1)$ & $\ell\geq 2$ &
$(1,0,\ldots,0)$ & $\ell{+}1$ & $(2,0,\ldots,0)$ & $\tfrac{(\ell+1)(\ell+2)}{2}$\\ 
\end{tabular}
\end{ruledtabular}
\end{table}
}

\begin{lemma}\label{lemma_selfdual_sym}
Let $\phi$ denote a faithful representation of a compact semisimple Lie algebra
$\fg$ such that $\phi$ is not self-dual and that the symmetric square $\Sym^2\phi$ is simple.
All possible cases (up to outer automorphisms of $\fg$) are given in Table~\ref{sym_square_nonselfdual}.
\end{lemma}

\begin{proof}
Due to  Theorem~\ref{simple} of Appendix~\ref{app_pre}, $\phi$  is simple and $\fg$ is simple.
Proposition~\ref{Malcev} of Appendix~\ref{app_pre} limits the possible cases to (a) $\fg=\su(\ell+1)$ with $\ell >1$, 
(b) $\fg=\so(4k+2)$ with $k \ge 1$, or (c) $\fg=\fe_6$.
Due to Lemma~\ref{simply_laced_sym} only the case (i) of Table~\ref{sym_square_nonselfdual} can appear
for (a). The irreducibility of the symmetric square in the case (i) is well known
(see p.~225 of Ref.~\onlinecite{FH91}). 
For (b), one can directly verify that $\Sym^2 \phi_{(0,0,0,1,0)}$ is not simple for $\fg=\so(10)$. 
We apply Propositions~\ref{prop_subord} and \ref{prop_part} of Appendix~\ref{Dynkin_techniques}
and the remaining case of $\so(6)\cong\su(4)$ 
has been already discussed. One can check by explicit computation that 
$\phi_{(1,0,0,0,0,0)}$,  $\phi_{(0,1,0,0,0,0)}$, and $\phi_{(0,0,0,0,0,1)}$ do not lead to a
simple symmetric square for
$\fe_6$ from (c), completing the proof.
\end{proof}

The case-by-case analysis of Theorem~\ref{sym_self} and Lemma~\ref{lemma_selfdual_sym}
imply Theorem~\ref{thm_sym_square}.

\section{Generalizing a Theorem of Coquereaux and Zuber\label{CZ}}
We provide here a straight-forward generalization of a result of 
Coquereaux and Zuber (see Proposition~\ref{one_norm_conjugated_1} below and 
Ref.~\onlinecite{CZ11}) from the case of two simple representations
of a simple, compact Lie algebra to the case of two semisimple representations
of a compact Lie algebra. Recall that $\psi_1 \boxtimes \psi_2$ denotes the
outer tensor product as defined in Section~\ref{def_alt_sym}.
We start by observing
a trivial property of the notation $\normone{\phi}:=\sum_{i\in\mathcal{I}} m_i$.
\begin{proposition} \label{one_norm_prop}
$\normone{\psi_1 \boxtimes \psi_2} = 
\normone{\psi_1}\normone{\psi_2}$ holds for 
two semisimple representations $\psi_1$ and $\psi_2$ of the compact Lie algebras
$\fg_1$ and $\fg_2$, respectively.
\end{proposition}

A recent result of Coquereaux and Zuber connects the values of
$\normone{\phi \otimes \psi}$ and $\normone{\phi \otimes \bar{\psi}}$.

\begin{proposition}[Thm.~1 of Ref.~\onlinecite{CZ11}]\label{one_norm_conjugated_1}
Given a compact simple Lie algebra $\fg$ and two simple
representations $\phi$ and $\psi$ of $\fg$, it follows that
$\normone{\phi \otimes \psi}=\normone{\phi \otimes \bar{\psi}}$.
\end{proposition}

Moreover, this is also valid for compact abelian Lie algebras as all their simple representations 
are one-dimensional.
\begin{proposition} \label{one_norm_abelian}
Let $\phi$ and $\psi$ denote simple representations of a compact 
abelian Lie algebra $\fg$, then 
$\normone{\phi \otimes \psi} = \normone{\phi \otimes \bar{\psi}}=1$.
\end{proposition}

One can now generalize the theorem of Coquereaux and Zuber to simple representations of compact Lie algebras.

\begin{proposition} \label{one_norm_conjugated_2}
Consider two simple representations $\phi$ and $\psi$ of 
a compact Lie algebra $\fg=\oplus_\ell \fg_\ell$ which can be decomposed
into its simple and abelian subalgebras~$\fg_\ell$.
One obtains $\phi\cong \phi_1 \boxtimes \phi_2 \boxtimes \cdots \boxtimes \phi_n$
and $\psi\cong\psi_1 \boxtimes \psi_2 \boxtimes \cdots \boxtimes \psi_n$
for  simple representations $\phi_\ell$ and $\psi_\ell$ of $\fg_\ell$. 
It follows that $\normone{\phi \otimes \psi}=\normone{\phi \otimes \bar{\psi}}$.
\end{proposition}
\begin{proof}
We note that $\phi \otimes \psi = (\boxtimes_\ell \phi_\ell) \otimes (\boxtimes_\ell \psi_\ell) = 
\boxtimes_\ell (\phi_\ell \otimes \psi_\ell) $ and
$\phi \otimes \bar{\psi} = (\boxtimes_\ell \phi_\ell) \otimes (\boxtimes_\ell \bar{\psi}_\ell) = 
\boxtimes_\ell (\phi_\ell \otimes \bar{\psi}_\ell) $. One concludes from 
Proposition~\ref{one_norm_prop} that
$\normone{\phi \otimes \psi}= \prod_\ell \normone{\phi_\ell \otimes \psi_\ell}$ and
$\normone{\phi \otimes \bar{\psi}}= \prod_\ell \normone{\phi_\ell \otimes \bar{\psi}_\ell}$ hold.
Applying Propositions \ref{one_norm_conjugated_1} and \ref{one_norm_abelian},
we get
$\normone{\phi_\ell \otimes \psi_\ell}= \normone{\phi_\ell \otimes \bar{\psi}_\ell}$
for each $\ell$. Therefore, we obtain
$\normone{\phi \otimes \psi}=\normone{\phi \otimes \bar{\psi}}$.
\end{proof}

All the previous propositions justify the following generalization of
Proposition~\ref{one_norm_conjugated_1}.

\begin{proposition}\label{one_norm_conjugated}
Given a compact Lie algebra $\fg$ and two semisimple 
representations $\phi$ and $\psi$ of $\fg$, it follows that
$\normone{\phi \otimes \psi}=\normone{\phi \otimes \bar{\psi}}$.
\end{proposition}
\begin{proof}
Consider the decompositions $\phi= \oplus_\ell \nu_\ell$ and $\psi=\oplus_\kappa \mu_\kappa$
of the representations  $\phi$ and $\psi$ into its simple components, where the different 
components $\nu_\ell$ and $\mu_\kappa$ are not necessarily distinct. 
Proposition~\ref{sumrules}(i) implies
$\normone{\phi \otimes \psi}=\normone{\oplus_{\ell, k} (\nu_\ell \otimes \mu_k)}
=\sum_{\ell, k}\normone{ \nu_\ell \otimes \mu_k}$ and
$\normone{\phi \otimes \bar{\psi}}=\normone{\oplus_{\ell,k} (\nu_\ell \otimes \bar{\mu}_k)}
=\sum_{\ell, k}\normone{ \nu_\ell \otimes \bar{\mu}_k}$.
We apply Proposition~\ref{one_norm_conjugated_2} and obtain
$\normone{ \nu_\ell \otimes \mu_k} =\normone{ \nu_\ell \otimes \bar{\mu}_k}$. Hence,
$\normone{\phi \otimes \psi}=\normone{\phi \otimes \bar{\psi}}$ also follows.
\end{proof}

\section{Details on the semisimple part\label{app_semisimple}}
The proof of Theorem~\ref{Thm_semsimple} will be stated in this appendix
after presenting two auxiliary results which follow from Proposition~\ref{one_two}. Recall
that a compact Lie algebra $\fg$ decomposes as $\fg =  \fs(\fg) \oplus \fc(\fg)$ 
where $\fs(\fg)$ and $\fc(\fg)$
denote its 
semsimple part and its center, respectively.
\begin{corollary}\label{compact_boxtimes_one}
Consider a subalgebra $\fh$ of a compact Lie algebra $\fg$ 
which observes
$\fh \cap \fc(\fg) = \fc(\fg)$, i.e., one has the decompositions  
$\fh = [\fh \cap \fs(\fg)] \oplus \fc(\fg)$ and $\fg = \fs(\fg) \oplus \fc(\fg)$. Note that $\fh \cap \fs(\fg)$ 
might not be semisimple. A  semisimple representation $\phi$ of $\fg$
decomposes as
$\phi \cong \oplus_{i \in \mathcal{I}, k \in \mathcal{K}}\,
[s_i \boxtimes c_k]^{\oplus m_{ik}}$ where 
$s_i$ (resp.\ $c_k$) denotes a simple representation of $\fs(\fg)$ (resp.\  $\fc(\fg)$).
The corresponding multiplicity of the simple representation $s_i \boxtimes c_k$ is given
by $m_{ik}$, and the classes of representations are indexed by 
$\mathcal{I}$ (resp.\ $\mathcal{K}$). Using the conditions
\begin{subequations}
\begin{align}
\label{condition_one}
& \normone{s_i|_{\fh \cap \fs(\fg)}} = \normone{s_i}=
1  \,\text{ holds for all }\, i\in \mathcal{I}
\,\text{ with }\, m_{ik}\neq 0 \,\text{ for some }\,
k\in \mathcal{K},\\
& (s_{i})|_{\fh \cap \fs(\fg)} 
\neq (s_{\ell})|_{\fh \cap \fs(\fg)} \,\text{ holds for all }\, i,\ell \in \mathcal{I} \nonumber \\ 
\label{condition_two}
& \phantom{(s_{i})|_{\fh \cap \fs(\fg)} 
\neq (s_{\ell})|_{\fh \cap \fs(\fg)}\,  }
\text{ with }\, i\neq \ell,\,
m_{ik}\neq 0,\, \text{and }\,  m_{\ell k}\neq 0 \,\text{ for some }\, k \in \mathcal{K},
\end{align}
\end{subequations}
it follows that   $\normtwo{\phi|_{\fh}} = \normtwo{\phi}$ 
$\Leftrightarrow$ $[\eqref{condition_one}\, \text{ and }\, \eqref{condition_two}]$ $\Leftrightarrow$
$[\normone{\phi|_{\fh}} = \normone{\phi}\, \text{ and }\, \eqref{condition_two}]$. 
\end{corollary}

\begin{proof}
The restriction of $\fg$ to $\fh$ is trivial for the 
simple representations $c_k$ as $\fh \cap \fc(\fg) = \fc(\fg)$. Therefore, the corollary follows from 
Proposition~\ref{one_two}(c).
\end{proof}

We apply now Corollary~\ref{compact_boxtimes_one} to 
(not necessarily simple) representations and obtain the following useful lemma.

\begin{lemma}\label{compact_boxtimes_two}
We retain the notation of Corollary~\ref{compact_boxtimes_one} and consider
a decomposition $\phi \cong \oplus_{r \in R, k \in \mathcal{K}}\, [r \boxtimes c_k]^{\oplus m_{rk}}$
where $R$ denotes a suitably chosen set of (not necessarily simple) representations
of $\fh \cap \fs(\fg)$
and $m_{rk} \in \{0,1,2,\ldots\}$. It follows that the condition $\normtwo{\phi|_{\fh}} = \normtwo{\phi}$ implies that 
$\normtwo{r|_{\fh \cap \fs(\fg)}}=\normtwo{r}$ holds for all $r\in R$ with $m_{rk} \neq 0$ for some $k \in \mathcal{K}$.
\end{lemma}

\begin{proof}
The condition $\normtwo{\phi|_{\fh}} = \normtwo{\phi}$ implies via 
Corollary~\ref{compact_boxtimes_one}
that \eqref{condition_one} and \eqref{condition_two} hold for suitable representations $s_i$. 
But the same $s_i$ appear
also in the decomposition of the representations $r$, and we apply Proposition~\ref{one_two}(c) to the subalgebra
$\fh \cap \fs(\fg)$ of $ \fs(\fg)$ to show that $\normtwo{r|_{\fh \cap \fs(\fg)}}=\normtwo{r}$.
\end{proof}

After these preparations, we can provide the details for the proof of Theorem~\ref{Thm_semsimple}.

\begin{proof}[Proof of Theorem~\ref{Thm_semsimple}]
We have to show that $\fs(\fh)=\fs(\fg)$ is implied by 
$\dim(\comm[(\phi \otimes \phi)|_{\fh}]) = \dim(\comm[\phi \otimes \phi])$.
But as $\fh = [\fh \cap \fs(\fg)] \oplus [\fh \cap \fc(\fg)] \subseteq 
\fs(\fg) \oplus [\fh \cap \fc(\fg)] \subseteq \fs(\fg) \oplus \fc(\fg)$ and 
$\dim(\comm[(\phi \otimes \phi)|_{\fh}]) \geq 
\dim(\comm[(\phi \otimes \phi)|_{\fs(\fg)\oplus [\fh \cap \fc(\fg)]}])
\geq \dim(\comm[\phi \otimes \phi])$ are valid, $\dim(\comm[(\phi \otimes \phi)|_{\fh}]) = 
\dim(\comm[(\phi \otimes \phi)|_{\fs(\fg)\oplus [\fh \cap \fc(\fg)]}])$ is a consequence of 
$\dim(\comm[(\phi \otimes \phi)|_{\fh}]) = \dim(\comm[\phi \otimes \phi])$. 
Thus, the proof specializes to the case when $\fh \cap \fc(\fg) = \fc(\fg)$
holds, i.e., $\fg = \fs(\fg) \oplus [\fh \cap \fc(\fg)]$; which we assume for the rest of the proof.
Using the notation of Corollary~\ref{compact_boxtimes_one}, the representation $\phi$ 
decomposes as $\phi \cong \oplus_{i \in \mathcal{I}, k \in \mathcal{K}}\,
[s_i \boxtimes c_k]^{\oplus m_{ik}}$. We apply the notation of Lemma~\ref{compact_boxtimes_two} and obtain 
$\phi \otimes \phi \cong \oplus_{r \in R, k \in \mathcal{K}}\, [r \boxtimes c_k]^{\oplus m_{rk}}$ where
$r \cong s_i \otimes s_j$ and $m_{rk}\in \{0,1,2,\ldots\}$. 
Lemma~\ref{compact_boxtimes_two} implies that $\normtwo{(s_i \otimes s_i)|_{\fh \cap \fs(\fg)}}
=\normtwo{s_i \otimes s_i}$ holds for each $i$.
Recall that any ideal of a semsimple Lie algebra is a direct factor (cf.\ 
Chap.~I, Sec.~6.4, Corollary to Prop.~5  of Ref.~\onlinecite{Bourb89}).
As $\phi$ is faithful, one finds for every
simple component $\fg_\ell$ of $\fs(\fg) \cong \oplus_{\ell} \fg_\ell$ an index $i$ such that
$\kernel(s_i) \cap \fg_\ell = \{0\}$ and $\kernel(s_i\otimes s_i) \cap \fg_\ell = \{0\}$.
Moreover, $\fs(\fg)/\kernel(s_i\otimes s_i)$ is semisimple and we are able to apply 
Theorem~\ref{thm:tensor_square} to the Lie algebra $\fs(\fg)/\kernel(s_i\otimes s_i)$ and its faithful 
representation $s_i\otimes s_i$. This implies that 
$\fh \cap [\fs(\fg)/\kernel(s_i\otimes s_i)] = \fs(\fg)/\kernel(s_i\otimes s_i)$ for each $i$.
In particular, we have that $\fh \cap \fg_\ell = \fg_\ell$ for each $\ell$ and the theorem follows.
\end{proof}

\section{Sharpening the bounds for the tensor-square theorem in the case of self-dual representations\label{sharp_app}}

In this appendix, we prove Proposition~\ref{bound_self_dual} which 
under the assumption that $\fg$ is simple and $\alpha$ is
self-dual
provides
bounds for the gap between the values of $ \normone{(\alpha \otimes \alpha)|_{\fh}}$ and
$\normone{\alpha \otimes \alpha}$ (as well as $ \normtwo{(\alpha \otimes \alpha)|_{\fh}}$ and  
$\normtwo{\alpha \otimes \alpha}$) for any
representation $\alpha$ of a compact  semisimple Lie algebra $\fg$ and its restriction to a
proper subalgebra $\fh$ of $\fg$. 
This generalizes inequalities given in Theorem \ref{thm:tensor_square}. First, let us recall a result by 
King and Wybourne\cite{KW96}.

\begin{proposition}[Prop.~4.2 of \onlinecite{KW96}]\label{prop:King}
For any compact simple Lie algebra $\fg$, the multiplicity of occurrence of
the adjoint representation $\theta$ in the tensor square of any self-dual and
simple representation $\alpha$ is given by  
the number $b(\alpha)$ of non-vanishing components in the highest weight $(\alpha_1, \ldots, \alpha_\ell)$ 
corresponding to $\alpha$.
\end{proposition}

In order to apply Proposition~\ref{prop:King}, we derive bounds which depend
on the multiplicity of the adjoint representation in the decomposition of  $\phi \otimes \bar{\phi}$.

\begin{lemma}\label{lem:gap_1}
Let $\phi$ denote a representation of the
compact semisimple Lie algebra $\fg$, and let $\fh$ denote a proper subalgebra of $\fg$. 
If  the multiplicity of the adjoint representation $\theta$ of $\fg$  in $\phi \otimes \bar{\phi}$ is $m$, then\\
(1) $\normone{(\phi \otimes \bar{\phi})|_{\fh}} \ge m +\normone{\phi \otimes \bar{\phi}}$,\\
(2)  $\normone{(\phi \otimes \phi)|_{\fh}} \ge m + \normone{\phi \otimes \phi}$,\\
(3) $\normtwo{(\phi \otimes \bar{\phi})|_{\fh}} \ge m^2 +\normtwo{\phi \otimes \bar{\phi}}$,\\
(4)  $\normtwo{(\phi \otimes \phi)|_{\fh}} \ge m^2 + \normtwo{\phi \otimes \phi}$,\\
(5) $\dim(\comm[(\phi \otimes \bar{\phi})|_{\fh}]) \ge m^2 + \dim(\comm[\phi \otimes \bar{\phi}])$,\\
(6) $\dim(\comm[(\phi \otimes \phi)|_{\fh}]) \ge m^2 +\dim(\comm[\phi \otimes \phi])$.
\end{lemma}

Note that Proposition~\ref{bound_self_dual} follows now by combining Proposition~\ref{prop:King} with 
Lemma~\ref{lem:gap_1}.

\begin{proof}[Proof of Lemma~\ref{lem:gap_1}]
We decompose $\phi \otimes \bar{\phi}$ into  $\theta^{\oplus m} \oplus \psi$ with maximal $m$.
Proposition~\ref{adjoint_one_norm_cond} implies that 
$\normone{\theta|_{\fh}} \ge \normone{\theta}+1$. Applying this inequality and  Proposition~\ref{sumrules}(i),
we obtain $\normone{(\phi \otimes \bar{\phi})|_{\fh}}= 
\normone{\theta^{\oplus m}|_{\fh}} + \normone{\psi|_{\fh}} =
m \normone{\theta|_{\fh}} + \normone{\psi|_{\fh}} \ge 
m (\normone{\theta}+1) + \normone{\psi} = m + \normone{\theta^{\oplus m}} + \normone{\psi}=
m +\normone{(\phi \otimes \bar{\phi})} $, 
which proves (1). Statement (2) follows from (1)
via Proposition~\ref{one_norm_conjugated}. 
Consider the decompositions $\theta\cong\oplus_{i} \nu_i^{\oplus m_i}$ and 
$\psi\cong \oplus_i \nu_i^{\oplus n_i}$
into simple representations $\nu_i$. It follows that
$\phi \otimes \bar{\phi}= \oplus_i \nu_i^{\oplus(m m_i + n_i)}$. 
Since $\normone{\theta|_{\fh}} \ge \normone{\theta} +1$, we get from 
Proposition \ref{one_two}(a)-(b) that 
$\normtwo{(\nu_p)|_{\fh}} \ge \normtwo{\nu_p} +1$ holds
for some $\nu_p$ with $m_p \ge 1$.  
Thus, one can write 
$\normtwo{\phi|_{\fh}}= \normtwo{\oplus_i (\nu_i)|_{\fh}^{\oplus(m m_i {+} n_i)}}
= \sum_i (m m_i {+} n_i)^2\normtwo{(\nu_i)|_{\fh}}= (m m_p {+} n_p)^2 
\normtwo{(\nu_p)|_{\fh}} + \sum_{i \ne p} (m m_i {+} n_i )^2\normtwo{(\nu_i)|_{\fh}}
\ge (m m_p {+} n_p)^2 
(\normtwo{\nu_p} +1) + \sum_{i \ne p} (m m_i {+} n_i )^2\normtwo{\nu_i} = 
(m m_p {+} n_p)^2 + \normtwo{\phi \otimes \bar{\phi}} \ge m^2 + \normtwo{\phi \otimes \bar{\phi}}$. 
This completes the proof of (3).
Statement~(4) follows via Proposition~\ref{two_norm_conjugated}. Obviously, (5) and (6) are a 
consequence of (3) and (4), respectively.
\end{proof}

\end{document}